\documentclass[11pt,letterpaper]{article}
\usepackage[margin=1in]{geometry}
\usepackage{rpmacros}
\usepackage[T1]{fontenc}
\usepackage[usenames,dvipsnames]{color}
\usepackage{amsmath}
\usepackage{stmaryrd}
\usepackage{lineno}
\usepackage{xfrac}
\usepackage{mathpazo}
\usepackage{mathtools}
\usepackage[colorinlistoftodos]{todonotes}
\usepackage{setspace}
\usepackage{nicefrac}
\usepackage{gitinfo}
\usepackage{float}
\usepackage{scrextend}
\usepackage{floatpag}
\usepackage{enumitem}
\usepackage[ruled, linesnumbered, vlined]{algorithm2e}
\let\oldnl\nl
\newcommand{\nonl}{\renewcommand{\nl}{\let\nl\oldnl}}

\usepackage[nottoc]{tocbibind}

\usepackage{tikz}
\usepackage{subcaption}
\usepackage{float}
\usepackage{import}
\allowdisplaybreaks

\usepackage{fullpage}
\usepackage[utf8]{inputenc}
\usepackage[russian,english]{babel}

\usepackage{color}
\usepackage[
pagebackref, 
pdfstartview=FitH,pdfpagemode=UseNone,colorlinks=true,citecolor=blue,linkcolor=blue]{hyperref}
\hypersetup{
  linkcolor=[rgb]{0.3,0.3,0.6},
  citecolor=[rgb]{0.2, 0.6, 0.2},
  urlcolor=[rgb]{0.6, 0.2, 0.2}
}

\usepackage{amsthm}
\usepackage{thmtools,thm-restate}
\usepackage[nameinlink,capitalise]{cleveref}

\numberwithin{equation}{section}
\declaretheoremstyle[bodyfont=\it,qed=\qedsymbol]{noproofstyle}

\declaretheorem[name=Observation,numbered=no]{observation*}

\declaretheorem[numberlike=equation]{subclaim}

\declaretheorem[numberlike=equation]{theorem}

\declaretheorem[name=Theorem,numbered=no]{theorem*}

\declaretheorem[numberlike=equation]{lemma}
\declaretheorem[name=Lemma,numbered=no]{lemma*}

\declaretheorem[numberlike=equation]{corollary}
\declaretheorem[name=Corollary,numbered=no]{corollary*}

\declaretheorem[name=Proposition,numbered=no]{proposition*}

\declaretheorem[name=Claim,numberlike=equation]{claim}
\declaretheorem[name=Claim,numbered=no]{claim*}

\declaretheorem[name=Conjecture,numbered=no]{conjecture*}

\declaretheorem[name=Question,numbered=no]{question*}

\declaretheoremstyle[bodyfont=\it,qed=$\lrcorner$]{defstyle} 

\declaretheorem[numberlike=equation,style=defstyle]{definition}
\declaretheorem[unnumbered,name=Definition,style=defstyle]{definition*}

\declaretheorem[unnumbered,name=Example,style=defstyle]{example*}

\declaretheorem[unnumbered,name=Notation=defstyle]{notation*}

\declaretheorem[unnumbered,name=Construction,style=defstyle]{construction*}

\declaretheorem[unnumbered,name=Remark,style=defstyle]{remark*}

\usepackage{nth}
\usepackage{intcalc}
\usepackage{etoolbox}
\usepackage{xstring}

\usepackage{ifpdf}
\ifpdf
\else
\usepackage[quadpoints=false]{hypdvips}
\fi

\newcommand{\ECCCurl}[2]{http://eccc.hpi-web.de/report/\ifnumcomp{#1}{>}{93}{19}{20}#1/#2/}

\newcommand{\shortECCC}[2]{\texttt{\href{http://eccc.hpi-web.de/report/\ifnumcomp{#1}{>}{93}{19}{20}#1/#2/}{eccc:TR#1-#2}}}

\newcommand{\parseECCC}[1]{
\StrSubstitute{#1}{TR}{}[\tmpstring]%
\IfSubStr{\tmpstring}{/}{ 
\StrBefore{\tmpstring}{/}[\ecccyear]%
\StrBehind{\tmpstring}{/}[\ecccreport]%
}{
\StrBefore{\tmpstring}{-}[\ecccyear]%
\StrBehind{\tmpstring}{-}[\ecccreport]%
}%
\shortECCC{\ecccyear}{\ecccreport}}

\usetikzlibrary{decorations.pathreplacing,arrows}

\newif\ifnote
\let\ifnote\iffalse

\ifnote
\newcommand{\phnote}[1]{\todo[color=red!100!green!33,
  size=\footnotesize]{ph: #1}}
\newcommand{\sriprahladhuvacha}[1]{\todo[color=red!100!green!33,inline,size=\small]{ph: #1}}

\newcommand{\ASnote}[1]{\textcolor{OliveGreen}{\guillemotleft Ashutosh: #1\guillemotright}}
\newcommand{\RPnote}[1]{\textcolor{BrickRed}{\guillemotleft RP: #1\guillemotright}}
\newcommand{\gitinfonotecolour}{Gray}
\else
\newcommand{\phnote}[1]{}
\newcommand{\sriprahladhuvacha}[1]{}
\newcommand{\TMnote}[1]{}
\newcommand{\ASnote}[1]{}
\newcommand{\RPnote}[1]{}
\newcommand{\gitinfonotecolour}{white}
\fi

\newcommand{\ehref}[1]{\href{mailto:#1}{\texttt{#1}}}

\renewcommand{\E}{\operatorname{\mathbb{E}}}
\newcommand{\AC}{\mathsf{AC}}

\let\epsilon\varepsilon

\onehalfspace

\SetKwComment{Comment}{/* }{ */}
\DontPrintSemicolon

\newcommand{\CDT}{\operatorname{CDT}}
\newcommand{\dom}{\operatorname{dom}}
\newcommand{\cor}{\operatorname{Cor}}
\newcommand{\stars}{\operatorname{stars}}
\newcommand{\calR}{\mathcal{R}}
\newcommand{\DT}{\mathsf{DT}}

\newcommand{\hastad}{H\aa stad}

\newcommand{\cleq}{\preccurlyeq}

\newcommand{\prob}[2]{\Pr_{#1}\left[#2\right]}

\newcommand{\zo}{\{0,1\}}

\newcommand{\bra}[1]{\left\{#1\right\}}

\newcommand{\cbra}[1]{\left(#1\right)}

\newcommand{\rhotilde}{\Tilde{\rho}}
\newcommand{\rhotildel}{\Tilde{\rho}_\ell}
\newcommand{\ptilde}{\Tilde{p}}
\newcommand{\Dtilde}{\Tilde{D}}
\newcommand{\calRtilde}{\Tilde{\calR}}

\newcommand{\dist}{\Tilde{\calR}_p(F)}

\newcommand{\bsigma}{\boldsymbol{\sigma}}
\newcommand{\bS}{\boldsymbol{S}}
\newcommand{\Stilde}{{\Tilde{S}}}
\newcommand{\bStilde}{\boldsymbol{\Tilde{S}}}

\newcommand{\bbeta}{\boldsymbol{\beta}}
\newcommand{\btau}{\boldsymbol{\tau}}
\newcommand{\balpha}{\boldsymbol{\alpha}}
\newcommand{\balphap}{\boldsymbol{\alpha'}}
\newcommand{\balphapp}{\boldsymbol{\alpha''}}
\newcommand{\bDtilde}{\boldsymbol{\Tilde{D}}}
\newcommand{\bGamma}{\boldsymbol{\Gamma}}

\newcommand{\brho}{\boldsymbol{\rho}}
\newcommand{\brhotilde}{\boldsymbol{\Tilde{\rho}}}

\newcommand{\calF}{\mathcal{F}}

\newcommand{\calA}{\mathcal{A}}
\newcommand{\calB}{\mathcal{B}}
\newcommand{\calC}{\mathcal{C}}
\newcommand{\calT}{\mathcal{T}}
\newcommand{\calE}{\mathcal{E}}

\newcommand{\size}{\operatorname{size}}
\newcommand{\depth}{\operatorname{depth}}
\newcommand{\ass}{\operatorname{assoc}}
\newcommand{\ZASS}{\operatorname{ASSOC}_0}

\newcommand{\bigo}[1]{O\left(#1\right)}

\newcommand{\gitinfonote}{git info:~\gitAbbrevHash\;,\;(\gitAuthorIsoDate)\; \;\gitVtag}

\title{Criticality of $\AC^0$-Formulae}
\date{}

\author{
    Prahladh Harsha
    \thanks{Tata Institute of Fundamental Research, Mumbai, India. Email: \ehref{prahladh@tifr.res.in}, \ehref{tulasimohanm@gmail.com}, \ehref{ashushankar98@gmail.com}. Research supported by the Department of Atomic
    Energy, Government of India, under project 12-R\&D-TFR-5.01-0500. Research of second author partially supported through the MATRICS grant MTR/2019/001226 of the Science and Engineering Research Board, Department of Science and Technology, Government of India.} 
\and
    Tulasimohan Molli
    \samethanks
\and 
    Ashutosh Shankar
    \samethanks
}

\begin{document}

\maketitle


{\let\thefootnote\relax
\footnotetext{\textcolor{\gitinfonotecolour}{\gitinfonote}
}}

\begin{abstract}
Rossman [In \emph{Proc.\ $34$th Comput.\ Complexity Conf.}, 2019] introduced the notion of \emph{criticality}. The criticality of a Boolean function $f \colon \zo^n \to \zo$ is the minimum $\lambda \geq 1$ such that for all positive integers $t$, 
\[ \Pr_{\brho \sim \calR_p}\left[\DT_{\depth}(f|_{\brho}) \geq t\right] \leq (p\lambda)^t. \]
\hastad's celebrated switching lemma shows that the criticality of any $k$-DNF is at most $O(k)$. Subsequent improvements to correlation bounds of $\AC^0$-circuits against parity showed that the criticality of any $\AC^0$-\emph{circuit} of size $S$ and depth $d+1$ is at most $O(\log S)^d$ and any \emph{regular} $\AC^0$-\emph{formula} of size $S$ and depth $d+1$ is at most $O(\frac1d \cdot \log S)^d$. We strengthen these results by showing that the criticality of \emph{any} $\AC^0$-formula (not necessarily regular) of size $S$ and depth $d+1$ is at most $O(\frac{\log S}{d})^d$, resolving a conjecture due to Rossman. 

This result also implies Rossman's optimal lower bound on the size of any depth-$d$ $\AC^0$-formula computing parity [Comput.\ Complexity, 27(2):209--223, 2018.]. Our result implies tight correlation bounds against parity, tight Fourier concentration results and improved \#SAT algorithm for $\AC^0$-formulae. 
\end{abstract}

\section{Introduction}

Understanding the power of various models of computation is the central goal of complexity theory. With respect to small-depth AND-OR circuits, the early works of Furst, Saxe and Sipser \cite{FurstSS1984}, Sipser \cite{Sipser1983-borel}, Ajtai \cite{Ajtai1983}, Yao \cite{Yao1985} and \hastad\  \cite{Hastad1989} using \emph{random restrictions} and Razborov \cite{Razborov1987} and Smolensky \cite{Smolensky1987} using the \emph{polynomial method} laid out a promising direction. 

Furst, Saxe and Sipser \cite{FurstSS1984} and Ajtai \cite{Ajtai1983} independently proved that the parity function requires super-polynomial sized constant depth AND-OR circuits to compute it. This was then later improved by Yao \cite{Yao1985} and \hastad\ \cite{Hastad1989} who proved that any depth-$(d+1)$, AND-OR circuit computing parity on $n$ bits requires size $2^{n^{\Theta(\nicefrac1d)}}$. The pièce de résistance of these results is the \emph{switching lemma} method introduced by Furst, Saxe and Sipser \cite{FurstSS1984}. Informally stated, it states that any $k$-DNF\footnote{A $k$-DNF is a Boolean formula in disjunctive normal form (DNF) in which each term has at most $k$ literals. A $k$-CNF is defined similarly.} reduces (aka \emph{switches}) to a low-width CNF with high probability, when acted upon by a $p$-random restriction. Very soon (in \hastad's paper itself \cite{Hastad1989}), it was discovered that it was more convenient and useful to state the switching lemma in terms of the depth of decision trees. This leads us to \hastad's switching lemma, one of the most celebrated theorems in theoretical computer science. Let $f$ be a $k$-DNF and $\calR_p$ denote the distribution of $p$-random restrictions ($p \in [0,1]$) where each variable independently is left unrestricted with probability $p$ and otherwise set uniformly to 0 or 1. Then,
\[ \prob{\rho \sim \calR_p}{\DT_{\text{depth}}(f|_\rho)\geq s} \leq (5pk)^s \ .\]

Despite the immense success of these methods in proving optimal lower bounds for small-depth AND-OR circuits and related models, they did not help in understanding limits of considerably stronger computational models. It was soon discovered that "other techniques" are needed to tackle these stronger models and this was made formal in the \emph{natural proof} approach by Razborov and Rudich \cite{RazborovR1997}. About a decade ago, interest in an \emph{improved} switching lemma was revived while trying to understand optimal correlation bounds of small depth AND-OR circuits with the parity function. While the early results of Ajtai \cite{Ajtai1983} were only able to show a correlation bound of $\exp(-\Omega(n^{1-\epsilon}))$, Beame, Impagliazzo and Srinivasan \cite{BeameIS2012} proved a considerably smaller correlation bound of $\exp(-\Omega(n/2^{2d(\log S)^{\nicefrac45}}))$ for depth-$d$ AND-OR circuits of size $S$ with the parity function over $n$ bits. This was then improved by Impagliazzo, Matthews and Paturi \cite{ImpagliazzoMP2012} and \hastad\ \cite{Hastad2014} who proved the optimal correlation bound of $\exp(-\Omega(n/(\log S)^d))$ for depth-$(d+1)$ AND-OR circuits of size $S$ with the $n$-bit parity function. \hastad\ proved this optimal correlation bound by proving the \emph{multi-switching lemma}, a significant strengthening of his earlier switching lemma. The multi-switching lemma is best described in terms of \emph{criticality}, a notion introduced subsequently by Rossman \cite{Rossman2019}.

The \emph{criticality} of a Boolean function $f \colon \zo^n \to \zo$ is the minimum $\lambda \geq 1$ such that 
\[ \Pr_{\rho \sim \calR_p}\left[\DT_{\depth}(f|_\rho) \geq s\right] \leq (p\lambda)^s. \]
Thus, \hastad's switching lemma in terms of criticality states that $k$-DNFs (and $k$-CNFs) have criticality $O(k)$. The Multi-switching lemma (in Rossman's reformulation in terms of criticality) states that a depth-$(d+1)$ AND-OR circuit of size $S$ has criticality $O(\log S)^d$. In the same paper, Rossman \cite{Rossman2019} showed a stronger result that depth $(d+1)$ \emph{regular} AND-OR \emph{formulae} of size $S$ have criticality $O\left(\frac{\log S}{d}\right)^d$, where regular means that all gates at the same height have equal fan-in. Parallel and independent of this line of work involving correlation bounds with parity, Rossman \cite{Rossman2018} showed that any depth-$(d+1)$ AND-OR formula (not necessarily regular) that computes the $n$-bit parity requires size at least $2^{\Omega(d(n^{\nicefrac1d}-1))}$. Our main result is a common strengthening (and unification) of all the above mentioned results, where we prove that any (not necessarily regular) depth $(d+1)$ AND-OR formula of size $S$ has criticality $O\left(\frac{\log S}{d}\right)^d$. More precisely, 

\begin{theorem}\label{thm:AC0criticality}
Let $F$ be an AND-OR formula of depth $d+1$ and size at most $S$, then for any $p \in [0,1]$ 
\[ \prob{\rho \sim \calR_p}{\DT_{\text{depth}}(F|_\rho) \geq s} \leq \left(p\cdot O\left(32^{d}\left(\frac{\log S}{d} + 1\right)^d\right)\right)^s.\]
\end{theorem}

As an immediate corollary of the above criticality result and \cite[Theorem 14]{Rossman2019}, we get the following results for general AND-OR formulae of depth $(d+1)$. Rossman had proved similar results for regular AND-OR formulae of depth $(d+1)$ \cite{Rossman2019}.

\begin{corollary}
Let $f\colon \zo^n\to \zo$ be computable by an AND-OR formula of depth $d+1$ and size at most $S$. Then
\begin{enumerate}
    \item Decision tree size bounds: $\DT_{\text{size}}(f) \leq O\left(2^{\left(1-\nicefrac1{O\left(\frac{1}{d}\log S\right)^d}\right)n}\right)$,
    \item Correlation bound with parity: $\cor(f,\oplus_n)\leq O\left(2^{-n/O\left(\frac{1}{d}\log S\right)^d}\right)$,
    \item Degree bounds:  $\prob{\rho \sim \calR_p}{\deg(F|_\rho) \geq s} \leq \left(p\cdot O\left(32^{d}\left(\frac{1}{d}\log S + 1\right)^d\right)\right)^s.$     
    \item $\ell_2$-Fourier concentration (Linial-Mansour-Nisan \cite{LinialMN1993}):  $\sum_{S \subseteq [n]\colon |S| \geq k} \widehat{f}(S)^2 \leq 2e \cdot e^{-k/O\left(\frac{1}{d}\log S\right)^d}$,
    \item $\ell_1$-Fourier concentration (Tal \cite{Tal2017}): $\sum_{S \subseteq [n]\colon |S| =k} \abs{\widehat{f}(S)} \leq O\left(\frac{1}{d}\log S\right)^{dk}$.
\end{enumerate}
\end{corollary} 

As indicated before, \cref{thm:AC0criticality} unifies (and arguably simplifies) all previous results for $\AC^0$ circuits and formulae in this context. It also yields satisfiability results for $\AC^0$-formulae along the lines of the Impagliazzo, Matthews and Paturi result \cite{ImpagliazzoMP2012} (see \cref{sec:satisfiability} for more details).

\subsection{Proof Overview}

The proof of \cref{thm:AC0criticality} is an adaptation of the proof of the multi-switching lemma due to \hastad\ \cite{Hastad2014} and Rossman's proof in the regular case \cite{Rossman2019}. As is typical in all proofs of the switching lemma, we construct a canonical decision tree (CDT) for the depth-$d$ formula. This CDT under a restriction is constructed in an inductive fashion by progressively refining the restriction. The main theorem is proved via the following statement (see \cref{lem:main} for the exact statement), which we prove inductively. For any $s$-length bitstring $a \in \zo^s$,
\[ \prob{\brho}{ \text{ There exists a path in }\CDT(F,\brho) \text{ labelled by the instruction-set } a \mid \brho \in \calT} \leq (p \cdot \lambda(F))^s,\]
where $\calT$ is any family of downward-closed set of restrictions and $\lambda(F)$ is the desired criticality bound that we wish to prove. 
The crucial difference from Rossman's proof in the regular setting is we prove the above statement subject to $\brho$ belonging to any downward-closed set. As in Rossman's proof, the event "There exists a path
...." is broken into several subevents, $\calE_t$ for $t$ ranging over a polynomially large set and each of the events $\calE_t$ is typically a conjunction of 3 events $\calA_t, \calB_t$ and $\calC_t$. The required probability can then bounded by an expression as follows:
\[ \sum_{t}\prob{}{ \calA_t \cap \calB_t \cap \calC_t \mid \calT} = \sum_t \prob{}{\calA_t \mid \calT}\cdot \prob{}{\calB_t \mid \calA_t \cap \calT}\cdot \prob{}{\calC_t \mid \calA_t \cap \calB_t \cap \calT}.\]
The advantage of using conditioning is that some of these intermediate probabilities (c.f., $\prob{}{\calC_t \mid \calA_t \cap \calB_t \cap\calT}$) can be bound using the inductive assumption provided the conditioned events are themselves downward-closed. The events $\calA_t$ and $\calB_t$ are chosen such that this is indeed the case. The sum over $t$ is then handled via convexity. The use of downward-closed sets to prove the inductive claim is inspired from \hastad's use of downward-closed sets in his proof of the multi-switching lemma \cite{Hastad2014}. However, the situation for depth-$d$ formulae is considerably more involved than the DNF/CNF setting and both the choice of the events as well as bounding these conditional probabilities require considerable care and subtlety (see \cref{sec:downward} and \cref{clm:hastad}). The use of downward-closed sets considerably simplifies the proof and yields an arguably simpler proof of the criticality bound, even in the regular setting \cite{Rossman2019}.

\subsection*{Organization}
The rest of the paper is organized as follows. We begin with some preliminaries in \cref{sec:prelims}, where we recall the standard notions of restrictions, decision trees and introduce variants of these notions such as restriction trees etc, which will be of use later. We then define canonical decision trees for depth-$d$ formulae in \cref{sec:CDT}, identical to the corresponding notion in the regular setting due to Rossman \cite{Rossman2019}. We then demonstrate the downward-closedness of some properties related to CDTs in \cref{sec:downward} and finally prove the main theorem in \cref{sec:criticality}. In \cref{sec:satisfiability}, we use the main lemma to give a randomized \#SAT algorithm for arbitrary $\AC^0$ formulae generalizing the corresponding algorithm due to Rossman in the regular setting \cite{Rossman2019}.

\section{Preliminaries}\label{sec:prelims}

For a positive integer $n \in \N$, $[n]$ refers to the set $\{1,2,\ldots,n\}$. All logarithms in this paper are to base 2. 

While studying distributions $D$ over some finite set $\Sigma$, we will use bold letters (i.e., $\bsigma$) to distinguish a random sample according to $D$ from a fixed element $\sigma \in \Sigma$.  Given any distribution $D$ on a finite set $\Sigma$, we let $\mu_D\colon S \to [0,1]$ denote the corresponding probability distribution (i.e., $\mu_D(\sigma)=\prob{\bsigma \sim D}{\bsigma = \sigma}$). We will drop the subscript $D$ from $\mu_D$ typically.

We begin by recalling the definition of an $\AC^0$-formula.

\begin{definition}[$\AC^0$-formulae] \label{def: AC0}
Let $d$ be a non-negative integer. Let $V$ be a set of variable indices. An \emph{$\AC^0$-formula} $F$ of depth-$d$ over variables $V$ is (inductively) defined as follows: A depth-$0$ formula is the constant $0$ or $1$ or a literal $x_v$ or $\neg x_v$ where $v$ is a variable index. For $d\geq 1$, a depth-$d$ $\AC^0$-formula is either an OR-formula or an AND-formula which are defined below. A depth-$d$ OR-formula is of the form $F_1 \vee F_2 \vee \dots \vee F_m$ where the $F_i$'s are either depth-$d'$ AND-formulae for some $1 \leq d' <d$ or depth-$0$ formulae. A depth-$d$ AND-formula $F = F_1 \wedge F_2 \dots \wedge F_m$ is defined similarly. 

Depth-$1$ AND-formulae and OR-formulae are usually referred to as terms and clauses respectively, while depth-$2$ AND-formulae and OR-formulae are called DNFs and CNFs respectively. 

The size of a formula is given by the number of depth-$1$ sub-formulae\footnote{Traditionally, the size is defined by the number of leaves or depth-$0$ formulas but for this paper, it would be more convenient to work with this definition.}. More precisely, $\size(F)$ is inductively defined as
\begin{align*}
    \size(F) := \begin{cases}
        0 & \text{if }\depth(F) = 0\\
        1 & \text{if }\depth(F) = 1\\
        \sum_{i=1}^m \size(F_i) & \text{if } F=F_1 \vee F_2 \vee \dots \vee F_m \text{ or } F_1 \wedge F_2 \dots \wedge F_m.
    \end{cases}
\end{align*}
The variable index set $V$ of a given formula $F$ unless otherwise specified is always assumed to be $[n]$.

We will sometimes identify a formula $F$ with the Boolean function it computes. We say "$F\equiv 1$" if this Boolean function is a tautology and "$F\equiv 0$" if it is a contradiction.  
\end{definition}

\subsection{Restrictions, Decision Trees and Restriction Trees}
We will be chiefly concerned with restrictions. 

\begin{definition}[restriction]
Given a variable index set $V$, a \emph{restriction} $\rho$ is a function $\rho \colon V \to \{0,1,*\}$ or equivalently a partial function from $V$ to $\{0,1\}$. We refer to the domain of this partial function as $\dom(\rho)$ and the remaining set of unrestricted variables, namely $V \setminus \dom(\rho)$, as $\stars(\rho)$. 

We say that two restrictions $\rho_1$ and $\rho_2$ are consistent if for every $v \in \dom(\rho_1)\cap \dom(\rho_2)$, we have $\rho(v_1)=\rho(v_2)$. 

We can define a partial ordering among restrictions as follows: we say $\rho_1 \cleq \rho_2$ if (1) $\stars(\rho_1) \subseteq \stars(\rho_2)$ and (2) $\rho_1$ and $\rho_2$ are consistent. In words, $\rho_1$ only ``sets more variables'' than $\rho_2$. Sometimes, we will only be interested in this order with respect to a particular subset $T$ of the variable index set $V$. In this case, we say \[\rho_1 \cleq_T \rho_2 \text{ if (1) } \stars(\rho_1)\cap T \subset \stars(\rho_2) \cap T \text{ and (2) $\rho_1$ and $\rho_2$ are consistent.}\] 

Given a formula $F$ and a restriction $\rho$, the restricted formula $F|_\rho$ refers to the formula obtained by relabeling literals involving variables indices in $\dom(\rho)$ according to $\rho$ (we make no further simplification to the formula). Given two consistent restrictions $\rho_1,\rho_2$, $F|_{\rho_1,\rho_2}$ refers to the formula $(F|_{\rho_1})|_{\rho_2}$ (which is identical  to $(F|_{\rho_2})|_{\rho_1}$).
\end{definition}

It will sometimes be convenient to consider an ordering among the variables in the domain of a restriction, especially when studying restrictions arising from decision trees.

\begin{definition}[ordered restriction]
    An \emph{ordered restriction} on a variable set $V$ is a sequence of the form $\alpha = (x_{v_1} \rightarrow b_1, \ldots, x_{v_t}\rightarrow b_t)$ where $t\in \N$, $b_i \in \set{0,1}$, and $v_1,\ldots, v_t$ are distinct elements of $V$. We will use $\dom(\alpha)$ to refer to the set $\set{v_1,\ldots, v_t}$. (We will typically use $\alpha$ or $\beta$ for ordered restrictions.)

    Any ordered restriction can be interpreted as a restriction $\rho$ with $\dom(\rho) = \set{v_1,\ldots, v_t}$. Similarly, given a restriction $\rho$ on $V$, and an ordering on $\dom(\rho)$, we have a natural representation of $\rho$ as an ordered restriction on $V$.
\end{definition}

The following is the standard definition of a decision tree except that we allow the internal nodes of the tree to have (out-)degree either $1$ or $2$.

\begin{definition}[decision tree]
\label{def:decision tree}
A \emph{decision tree} is a a finite rooted binary tree where 
\begin{itemize}\itemsep0pt
    \item each internal node is labelled by a variable, has one or two children and the edges to its children have distinct labels from the set $\{0,1\}$, 
    \item the leaves are labelled by $0$ or $1$, and
    \item the variables appearing in any root-to-leaf path are distinct.
\end{itemize}
For each node $v$ (including leaf node), the root-to-node path in the decision tree naturally corresponds to an ordered restriction, which we denote by $\alpha^\Gamma_v$ (this restriction is non-trivial for every non-root node). 

The \emph{depth} of a decision tree $T$, denoted by $\depth(T)$, is defined as the maximum number of degree-$2$ nodes along any root-to leaf path in $T$. Note that this may be shorter than the length of the corresponding ordered restriction, which includes the degree-$1$ nodes also.

A decision tree is said to compute a Boolean function $F\colon \zo^V\to \zo$ under a restriction $\rho$ if the following conditions hold
\begin{itemize}
    \item any internal vertex labelled by a variable index, say $v$, that is in $\dom(\rho)$ has degree one, with the edge to the only child labelled with $\rho(v)$,
    \item any internal vertex labelled by a variable index in $\stars(\rho)$ has degree two and
    \item for every leaf $v$, we have $F|_{\rho,\alpha_v} \equiv \text{label}(v)$.\qedhere
\end{itemize}
\end{definition}

An "honest-to-god" decision tree (with all internal nodes having degree two) can be obtained from the above decision tree by contracting the degree-1 edges. However, we will find it convenient to keep this information about degree-1 nodes while constructing decision trees for functions under a restriction. Note that if a decision tree $T$ computes a function $F$ under the restriction $\rho$, then the contracted decision tree $T'$ computes the function $F|_\rho$.

To prove the criticality bound for a given formula $F$, we construct a canonical decision tree ($\CDT$) for $F$ under a (random) restriction $\rho$. This $\CDT$ is constructed in an inductive fashion by constructing the $\CDT$'s for $F$'s sub-formulae first and then using  these $\CDT$'s to construct $F$'s $\CDT$. While doing so, we progressively refine the restriction so that the final restriction under which the $\CDT$ is constructed is the target restriction $\rho$. This naturally leads us to the notion of \emph{restriction trees}, which is essentially a family of restrictions, one for each sub-formula of a given formula, such that the restrictions get refined as we move from child to parent in the formula tree.

\begin{definition}[restriction tree]\label{def:tree of restrictions}
Let $F$ be a formula on the variable index set $V$ and $T_F$ the set of all sub-formulae of $F$. The elements of $T_F$ have a natural bijection with the underlying formula tree of $F$. A restriction tree for $F$, denoted by $\rhotilde$, associates a restriction with each node in $T_F$, formally $\rhotilde\colon T_F \to \bra{0,1,*}^V$, such that for $G, H \in T_F$ where $G$ is a sub-formula of $H$, we have $\rhotilde(H) \cleq \rhotilde(G)$. In other words, the sequence of restrictions on any leaf-to-root path sets increasingly more variables as we approach the root. 

For any sub-formula $G$ of $F$, we let $\rhotilde|_G$ denote the restriction of $\rhotilde$ to the set $T_G$ of sub-formulae of $G$.
\end{definition}

We will use the "tilde" notation to distinguish between restrictions $\rho$ and restriction trees $\rhotilde$. Observe that, by definition, every restriction $\rho$ in a restriction tree $\rhotilde$ corresponding to a formula $F$ satisfies $\rho \cleq \rhotilde(F)$ and are hence consistent with each other. 

\subsection{Representation of restrictions and restriction trees}\label{sec:Sampleable}

Recall that a restriction $\rho$ is a partial function from the set  $V$ of variables to $\zo$. Sometimes (especially when dealing with \emph{random restrictions}), it will be convenient to work with a (redundant) representation of $\rho$ given by the pair $(\sigma, S)$ where $\sigma \colon V \to \zo$ is a \emph{global} assignment consistent with $\rho$ and $S = \stars(\rho)$. Note this is representation is redundant as we only need $\sigma|_{\overline{S}}$ to specify $\rho$. When sampling restrictions, it will be easier to sample the pair $(\sigma, S)$ from some distribution and set $\rho := \rho_{(\sigma, S)}$ to be the restriction given by 
\[\rho_{(\sigma,S)}(v) = \begin{cases}
        \sigma(v) & \text{ if } v \notin S,\\
        * & \text{ if } v \in S.
\end{cases}\] 

This representation naturally extends to restriction trees $\rhotilde \colon T_F \to \bra{0,1,*}$ which are given by a pair $(\sigma, \Stilde)$ where $\sigma \colon V \to \zo$ is a global assignment consistent with all the restrictions in the restriction tree and $\Stilde \colon T_F \to 2^V$ is defined as $\Stilde(G) := \stars(\rhotilde(G))$. Notice that any $\Stilde$ satisfies the monotonicity property that if $G$ is a sub-formula of $H$ in $T_F$, we have $\Stilde(H) \subseteq \Stilde(G)$. Given any such $\Stilde$ that satisfies the monotonicity property and a global assignment $\sigma \colon V \to \zo$, the corresponding restriction tree $\rhotilde_{(\sigma,\Stilde)}$ is given by $\rhotilde_{(\sigma,\Stilde)}(G) := \rho_{(\sigma, \Stilde(G))}$ for all $G \in T_F$.

\section{Canonical decision tree}\label{sec:CDT}

In this section, we construct a canonical decision tree (CDT) for a formula $F$. This definition is identical to Rossman's definition \cite[Definition 19]{Rossman2019} (except that Rossman defines it completely in terms of ordered restrictions while we define it using decision trees which have both degree-1 and degree-2 internal nodes).

Let us first recall the CDT construction for DNFs in the proof of \hastad's classical switching lemma \cite{Beame1994,Razborov1995,Hastad2014}. Let $F= T_1 \vee \dots \vee T_m$ be a DNF and $\rho$ a restriction on the variables of $F$. To construct $\CDT(F,\rho)$ we do the following: 
\begin{enumerate}
    \item Find the first term $T$ (from left to right), not forced to 0 by $\rho$. If there is no such term, return the tree comprising of a single leaf node labelled 0.
    \item If $T|_\rho \equiv 1$, return the tree comprising of a single leaf node labelled 1.
    \item Let $Y$ be the set of $\rho$-unrestricted variables in $T$. Let $\Gamma$ be the $\CDT(T,\rho)$ constructed from the complete balanced binary tree of depth $|Y|$ indexed by the variables of $Y$ and labelling the $2^{|Y|}$ appropriately. \label{item:balancing}
    \item For each leaf $v$ of $\Gamma$, inductively replace $v$ with $\CDT(F|_{\alpha_v},\rho)$ where $\alpha_v$ is the (ordered) restriction corresponding to leaf $v$.
\end{enumerate}

The construction of CDTs for depth-$d$ formulae will be inspired by the above CDT construction for DNFs. Note that in Step~\ref{item:balancing}, we used a complete binary tree instead of the best decision tree for the term $T$ (see \cref{fig: term CDT}). The rationale for doing this is because while proving the switching lemma, we wanted to attribute a 0-leaf in $\Gamma$ to a 1-leaf which shares the same set of variables. We will need a similar property in our construction. To this end, we perform a balancing operation which ensures that every 0-leaf has a corresponding 1-leaf such that the two associated ordered restrictions share the same set of variables (this is the 0-balancing operation defined below. The 1-balancing operation is similar with the roles of 0 and 1 reversed). 

\begin{figure}
    \centering
    \begin{subfigure}[h]{0.4\textwidth}
        \centering
        \begin{tikzpicture}[scale = 0.3, treenode/.style={circle, draw, thick}, leaf/.style={rectangle, draw, thick}]
        
        \node[treenode] at (2, 9) (x1){$x_1$};
        \node[leaf] at (0, 6) (leaf1) {0};
        \node[treenode] at (4, 6) (x2){$x_2$};
        \node[leaf] at (2, 3) (leaf2){0};
        \node[treenode] at (6, 3) (x3){$x_3$};
        \node[leaf] at (4, 0) (leaf3){0};
        \node[leaf] at (8, 0) (leaf4){1};
        
        \draw (x1) -- (leaf1);
        \draw (x1) -- (x2);
        \draw (x2) -- (leaf2);
        \draw (x2) -- (x3);
        \draw (x3) -- (leaf3);
        \draw (x3) -- (leaf4);   
        \end{tikzpicture}
        \caption{Optimal $\DT$ for $x_1\wedge x_2 \wedge x_3$}
    \end{subfigure}
    \hfill
    \begin{subfigure}[h]{0.4\textwidth}
        \centering
        \begin{tikzpicture}[scale = 0.3, treenode/.style={circle, draw, thick}, leaf/.style={rectangle, draw, thick}]
        
        \node[treenode] at (9, 9) (x1){$x_1$};

        \node[treenode] at (5, 6) (x21) {$x_2$};
        \node[treenode] at (13, 6) (x22){$x_2$};

        \node[treenode] at (3, 3) (x31){$x_3$};
        \node[treenode] at (7, 3) (x32){$x_3$};
        \node[treenode] at (11, 3) (x33){$x_3$};
        \node[treenode] at (15, 3) (x34){$x_3$};

        \node[leaf] at (2, 0) (leaf1){0};
        \node[leaf] at (4, 0) (leaf2){0};
        \node[leaf] at (6, 0) (leaf3){0};
        \node[leaf] at (8, 0) (leaf4){0};
        \node[leaf] at (10, 0) (leaf5){0};
        \node[leaf] at (12, 0) (leaf6){0};
        \node[leaf] at (14, 0) (leaf7){0};     
        \node[leaf] at (16, 0) (leaf8){1};
        
        \draw (x1) -- (x21);
        \draw (x1) -- (x22);

        \draw (x21) -- (x31);
        \draw (x21) -- (x32);
        \draw (x22) -- (x33);
        \draw (x22) -- (x34);
        
        \draw (x31) -- (leaf1);
        \draw (x31) -- (leaf2);
        \draw (x32) -- (leaf3);
        \draw (x32) -- (leaf4); 
        \draw (x33) -- (leaf5);
        \draw (x33) -- (leaf6); 
        \draw (x34) -- (leaf7);
        \draw (x34) -- (leaf8);      
        \end{tikzpicture}
        \caption{Completed balanced $\DT$ for $x_1 \wedge x_2 \wedge x_3$}
    \end{subfigure}
    \caption{Illustration of $\DT(x_1\wedge x_2\wedge x_3)$ used in the $\CDT$ construction in the proof of \hastad's Switching Lemma.}
    \label{fig: term CDT}
\end{figure}

\subsection{0-Balancing and 1-Balancing}
    
Given a decision tree $\Gamma$ for a Boolean function $F$, the 0-balanced version $\Gamma'$ is constructed as follows. We first pull-up the zeros, in other words, if there is any subtree all of whose leaves are labelled 0, we contract the entire subtree to a single leaf node labelled 0. The construction then proceeds in $d$ rounds where $d$ is the length of the longest root-to-leaf path in $\Gamma$ (note this is not necessarily the depth of $\Gamma$ due to the presence of degree-1 nodes). This process leaves the 1-leaves in $\Gamma$ unaltered. As we proceed, we also construct a map $\ass$ which associates each leaf (both 0 and 1 leaves) in $\Gamma'$ with a 1-leaf in $\Gamma'$. To begin with, this map $\ass$ associates each 1-leaf to itself (i.e, if $u$ is a 1-leaf, then $\ass(u)=u$).

In the $i^{th}$ round, we consider all 0-leaves in $\Gamma$ at distance $(d-i)$ from the root. Let $u$ be one such 0-leaf and $T_u$ the subtree rooted at the sibling of $u$. Observe that $T_u$ necessarily has some leaf labelled 1, else the entire subtree rooted at the parent of $u$ would have been contracted to a single leaf node labelled 0. We then mirror the entire subtree $T_u$ at the leaf node $u$ and relabel all the leaves of this mirrored subtree with 0. These are the 0-leaves of $\Gamma'$. For each such newly created 0-leaf $w$ (in the mirrored subtree $T_u$), let $w'$ be the corresponding leaf in the tree $T_u$. Set $\ass(w) \gets \ass(w')$. 

See \cref{fig:0 balancing} for an illustration of the 0-balancing process. Observe that if we 0-balance the best decision tree for a term, we obtain the complete balanced tree (see \cref{fig: term CDT}).


\begin{figure}
    \centering
    \begin{subfigure}[h]{0.15\textwidth}
        \begin{tikzpicture}[scale = 0.3, treenode/.style={circle, draw, thick}, leaf/.style={rectangle, draw, thick}]
        
        \node[treenode] at (10, 10) (x1){$x_1$};
        \node[leaf] at (7, 6) (collapsed0) {0};
        \node[treenode] at (13, 6) (x5){$x_5$};
        \node[treenode] at (11, 3) (x2){$x_2$};
        \node[leaf] at (10, 0) (leaf4){0};
        \node[leaf] at (12, 0) (leaf5){1};
        \node[leaf] at (15, 3) (leaf6){0};
        
        \draw (x1) -- (collapsed0);
        \draw (x1) -- (x5);
        \draw (x5) -- (x2);
        \draw (x2) -- (leaf4);
        \draw (x2) -- (leaf5);
        \draw (x5) -- (leaf6);
        
        \end{tikzpicture}
        \caption{initial}
        \end{subfigure}
        \hfill
        \begin{subfigure}[h]{0.15\textwidth}
        \begin{tikzpicture}[scale = 0.3, treenode/.style={circle, draw, thick}, leaf/.style={rectangle, draw, thick}]
    
        \node[treenode] at (10, 10) (x1){$x_1$};
        \node[leaf] at (7, 6) (x3){0};
        \node[treenode] at (13, 6) (x5){$x_5$};
        \node[treenode] at (11, 3) (x2){$x_2$};
        \node[leaf] at (10, 0) (leaf4){0};
        \node[leaf] at (12, 0) (leaf5){1};
        \node[leaf] at (15, 3) (leaf6){0};

        \node at (8, 0) {$w_1$};
        \node at (14, 0) {$w_2$};
    
        \draw (x1) -- (x3);
        \draw (x1) -- (x5);
        \draw (x5) -- (x2);
        \draw (x2) -- (leaf4);
        \draw (x2) -- (leaf5);
        \draw (x5) -- (leaf6);
    
        \draw[->, blue, thick] (leaf4) .. controls (11, -2) .. (leaf5); 
    
        \end{tikzpicture}
        \caption{first round \\ \textcolor{blue}{$assoc(w_1) \gets w_2$}}
        \end{subfigure}
        \hfill
        \begin{subfigure}[h]{0.15\textwidth}
        \begin{tikzpicture}[scale = 0.3, treenode/.style={circle, draw, thick}, leaf/.style={rectangle, draw, thick}]
    
        \node[treenode] at (10, 10) (x1){$x_1$};
        \node[leaf] at (7, 6) (x3){0};
        \node[treenode] at (13, 6) (x5){$x_5$};
        \node[treenode] at (11, 3) (x2){$x_2$};
        \node[leaf] at (10, 0) (leaf4){0};
        \node[leaf] at (12, 0) (leaf5){1};
    
        \node[treenode] at (15, 3) (newx2){$x_2$};
        \node[leaf] at (14, 0) (leaf6){0};
        \node[leaf] at (16, 0) (leaf7){0};
    
        \node at (14, -1.5) {$w_3$};
        
        \draw (x1) -- (x3);
        \draw (x1) -- (x5);
        \draw (x5) -- (x2);
        \draw (x5) -- (newx2);
        \draw (x2) -- (leaf4);
        \draw (x2) -- (leaf5);
    
        \draw (newx2) -- (leaf6);
        \draw (newx2) -- (leaf7);
    
        \draw[->, blue, thick] (leaf4) .. controls (11, -2) .. (leaf5); 
    
        \draw[->, red, thick] (leaf6) .. controls (12, -3) .. (leaf4);
        \draw[->, red, thick] (leaf7) .. controls (14, -3) .. (leaf5);
    
        \end{tikzpicture}
        \caption{second round \\ \textcolor{red}{$assoc(w_3) \gets$} \textcolor{blue}{$assoc(w_1) \gets w_2$}}
        \end{subfigure} 
        \hfill
        \begin{subfigure}[h]{0.3\textwidth}
        \begin{tikzpicture}[scale = 0.3, treenode/.style={circle, draw, thick}, leaf/.style={rectangle, draw, thick}]
            
        \node[treenode] at (10, 10) (x1){$x_1$};
        \node[treenode] at (6, 6) (newx5){$x_5$};
        \node[treenode] at (14, 6) (x5){$x_5$};
        \node[treenode] at (12, 3) (x2){$x_2$};
        \node[leaf] at (11, 0) (leaf4){0};
        \node[leaf] at (13, 0) (leaf5){1};
    
        \node[treenode] at (16, 3) (newx2){$x_2$};
        \node[leaf] at (15, 0) (leaf6){0};
        \node[leaf] at (17, 0) (leaf7){0};

        \node[treenode] at (4, 3) (newx2l){$x_2$};
        \node[treenode] at (8, 3) (newx2r){$x_2$};
        \node[leaf] at (3, 0) (leaf1){0};
        \node[leaf] at (5, 0) (leaf2){0};
        \node[leaf] at (7, 0) (leaf3){0};
        \node[leaf] at (9, 0) (leaf9){0};

        \node at (7, -1.5) {$w_7$};

        \draw (x1) -- (newx5);
        \draw (x1) -- (x5);
        \draw (x5) -- (x2);
        \draw (x5) -- (newx2);
        \draw (x2) -- (leaf4);
        \draw (x2) -- (leaf5);
    
        \draw (newx2) -- (leaf6);
        \draw (newx2) -- (leaf7);
    
        \draw[->, blue, thick] (leaf4) .. controls (12, -2) .. (leaf5); 
    
        \draw[->, red, thick] (leaf6) .. controls (13, -3) .. (leaf4);
        \draw[->, red, thick] (leaf7) .. controls (15, -3) .. (leaf5);
    
        \draw (newx5) -- (newx2l);
        \draw (newx5) -- (newx2r);
        \draw (newx2l) -- (leaf1);
        \draw (newx2l) -- (leaf2);
        \draw (newx2r) -- (leaf3);
        \draw (newx2r) -- (leaf9);
    
        \draw[->, olive, thick] (leaf1) .. controls (7, -4) .. (leaf4);
        \draw[->, olive, thick] (leaf2) .. controls (9, -4) .. (leaf5);
        \draw[->, olive, thick] (leaf3) .. controls (11, -4) .. (leaf6);
        \draw[->, olive, thick] (leaf9) .. controls (13, -4) .. (leaf7);
    
        \end{tikzpicture}
        \caption{third round \\ \textcolor{olive}{$assoc(w_7) \gets$} \textcolor{red}{$assoc(w_3) \gets$} \textcolor{blue}{$assoc(w_1) \gets w_2$}}     
        \end{subfigure}
        \caption{Illustration of 0-balancing process}
        \label{fig:0 balancing}
\end{figure}

At the end of this process, observe that $\Gamma$ is transformed into another decision tree $\Gamma'$ such that the following hold.
\begin{itemize}
    \item If $\Gamma$ computes a function $F$ under some restriction $\rho$, so does $\Gamma'$. 
    \item The 1-leaves in $\Gamma'$ are in 1-1 correspondence with the 1-leaves in $\Gamma$. Furthermore, the two 1-leaves (the one in $\Gamma$ and its associated 1-leaf in $\Gamma'$) correspond to identical ordered restrictions.
    \item Every 0-leaf $w$ in $\Gamma'$ has an associated 1-leaf in $\Gamma'$ given by $\ass(w)$. Furthermore, the corresponding ordered restrictions (namely $\alpha^{\Gamma'}_w$ and $\alpha^{\Gamma'}_{\ass(w)}$) share the same set of variables which are queried in the same order along both these root-to-leaf paths.
\end{itemize}

Let us now try to understand what are the 0-leaves constructed in the 0-balancing process. Let $w$ be any 0-leaf in $\Gamma'$ and $w'=\ass(w)$ be the corresponding 1-leaf. Furthermore, let $\alpha := \alpha^{\Gamma'}_{w} = ({v_1} \mapsto c_1, \dots , {v_t} \mapsto c_t)$ and $\beta := \alpha^{\Gamma'}_{w'} = ({v_1} \mapsto d_1, \dots , {v_t} \mapsto d_t)$. First, we must have that $\dom(\alpha)=\dom(\beta)$ and that the variables in this common domain must be queried in the same order. Furthermore, whenever $\alpha$ differs from $\beta$, the ordered restriction formed by following $\beta$ up to the step prior to this particular point of disagreement and then taking a step according to $\alpha$ must cause the formula $F|_\rho$ to evaluate to 0. This occurs as the mirroring operation is performed only at such nodes. More precisely, let $\alpha_i$ denote the ordered restriction  $({v_1} \mapsto d_1, {v_2} \mapsto d_2, \dots , {v_{i-1}} \mapsto d_{i-1}, {v_i} \mapsto c_i)$. Note, $\alpha_i$ is the ordered restriction of length $i$ which is identical to $\beta$ in the first $i-1$ variables and then is similar to $\alpha$ on the $i^{th}$ variable. The ordered restrictions $\alpha$ and $\beta$ satisfy the following:
\[ \forall i \in [t], c_i \neq d_i \implies F|_{\rho,\alpha_i} \equiv 0. \]
Furthermore, the converse also holds. That is, let $\beta$ corresponds to an ordered restriction of some 1-leaf in $\Gamma'$, then the ordered restriction $\alpha$ (with the same domain and same order of querying) corresponds to a 0-node in $\Gamma'$ only if the above condition holds.

Since this is an important point, we summarize the above discussion in the following definition and claim.

\begin{definition}\label{def:pivot}
Let $F$ be a Boolean function, $\rho$ a restriction and $\alpha = ({v_1} \mapsto c_1, \dots , {v_t} \mapsto c_t),\beta = ({v_1} \mapsto d_1, \dots , {v_t} \mapsto d_t)$ be two ordered restrictions (on the same domain and order of querying). 
        We say $\alpha \in \ZASS(F,\rho,\beta)$ iff 
        \[ \forall i \in [t], c_i \neq d_i \implies F|_{\rho,\alpha_i} \equiv 0 \]where $\alpha_i$ refers to the ordered restriction $({v_1} \mapsto d_1, {v_2} \mapsto d_2, \dots , {v_{i-1}} \mapsto d_{i-1}, {v_i} \mapsto c_i)$.
\end{definition}
    
\begin{claim}\label{claim:pivot claim}
Let $\Gamma$ compute the formula $F$ under the restriction $\rho$ and $\Gamma'$ be the 0-balanced version of $\Gamma$. Let $w$ be a 1-leaf in $\Gamma$ (and hence also $\Gamma'$) and $\beta$ be the corresponding ordered restriction. Then $\alpha$ is an ordered restriction corresponding to a 0-leaf $w'$ with $\ass(w') =w$ iff $\alpha \in \ZASS(F,\rho, \beta)$.
\end{claim}

1-balancing is defined similarly with the roles of 0 and 1 reversed

\subsection{CDT Definition}
We are now ready to define the canonical decision tree (CDT). As indicated before, this definition is identical to \cite[Definition 19]{Rossman2019}.

\begin{definition}
    \label{def:CDT AC0}
        Given a formula $F$ on variable set $V$  and associated restriction tree $\rhotilde\colon T_F \to \bra{0,1,*}^V$, we define the canonical decision tree, denoted by $\CDT(F,\rhotilde)$, inductively (on depth and the number of variables) as follows:
    \begin{enumerate}
        \item  If $F$ is a constant 0 or 1, then $\CDT(F,\rhotilde)$ is the unique tree with a single leaf node labelled by the appropriate constant.
        \item If $F$ is a literal $x$ or $\neg x$, then 
        \begin{itemize}
            \item if $x$ is set by $\rhotilde(F)$ to a constant, then
            $\CDT(F,\rhotilde)$ is the unique tree with a single node labelled by the appropriate constant.
            \item Otherwise if $x$ is unset by $\rhotilde(F)$, then $\CDT(F,\rhotilde)$ is the tree with 3 nodes where the root is labelled by $x$ and the two children are labelled appropriately by 0 or 1.
        \end{itemize}
        \item If $F = F_1 \vee \dots \vee F_m$, then 
        \begin{itemize}
            \item If $F_1|_{\rhotilde(F)} \equiv F_2|_{\rhotilde(F)} \equiv \dots \equiv F_{m}|_{\rhotilde(F)} \equiv 0$, then $\CDT(F,\rhotilde)$ is the unique tree with a single leaf node labelled 0. 
            \item 
            Else, there is some $1 \leq \ell \leq m$ such that $F_1|_{\rhotilde(F)} \equiv \dots \equiv F_{\ell-1}|_{\rhotilde(F)} \equiv 0$ and $F_{\ell}|_{\rhotilde(F)} \not\equiv$ 0.
            \item If $F_{\ell}|_{\rhotilde(F)} \equiv 1$, then $\CDT(F,\rhotilde)$ is the unique tree node with a single leaf node labelled 1. 
            \item If $F_{\ell}|_{\rhotilde(F)} \not\equiv$ constant, then do the following steps to construct $\CDT(F,\rhotilde)$
            \begin{enumerate}
                \item \label{step3a} Let $\Gamma$ be $\CDT(F_\ell,\rhotilde|_{{F_\ell}})$ constructed inductively (since $\depth(F_\ell)< \depth(F)$).
                \item \label{step3b} Apply the restriction $\rhotilde(F)$ to $\Gamma$ to get $\Gamma'$ and remove all the sub-trees which are inconsistent with $\rhotilde(F)$\footnote{This step introduces degree 1 nodes in the decision tree.}.
                \item \label{step3c} 0-balance $\Gamma'$ to get $\Gamma"$.
                \item For each 0-leaf $u$ of $\Gamma"$, replace $u$ by $\CDT(F|_{\alpha_u}, \rhotilde)$ where $\alpha_u$ is the ordered restriction corresponding to $u$ in $\Gamma"$.    
            \end{enumerate}
        \end{itemize}
        The case when $F = F_1 \wedge \dots \wedge F_m $ is a conjunction of sub-formulas is handled similarly (with the roles of 0 and 1 reversed). \qedhere
    \end{enumerate}
\end{definition}

\begin{figure}
    \centering
    \hspace*{-1cm}
    \makebox[1.2\textwidth][c]{
    \begin{subfigure}[h]{0.15\textwidth}
        \centering
        \begin{tikzpicture}[scale = 0.3, treenode/.style={circle, draw, thick}, leaf/.style={rectangle, draw, thick}]
            
        \node[treenode] at (10, 10) (x1){$x_1$};
        \node[treenode] at (6, 6) (x3){$x_3$};
        \node[treenode] at (14, 6) (x5){$x_5$};

        \node[leaf] at (4, 3) (leaf1){0};
        \node[treenode] at (8, 3) (x4){$x_4$};

        \node[treenode] at (12, 3) (x2){$x_2$};
        \node[leaf] at (16, 3) (leaf6){0};

        \node[leaf] at (7, 0) (leaf2){0};
        \node[leaf] at (9, 0) (leaf3){1};

        \node[leaf] at (11, 0) (leaf4){0};
        \node[leaf] at (13, 0) (leaf5){1};
    
        \draw (x1) -- (newx5);
        \draw (x3) -- (leaf1);
        \draw (x3) -- (x4);
        \draw (x4) -- (leaf2);
        \draw (x4) -- (leaf3);
        \draw (x1) -- (x5);
        \draw (x5) -- (x2);
        \draw (x2) -- (leaf4);
        \draw (x2) -- (leaf5);
        \draw (x5) -- (leaf6);

        \draw[->, very thick] (17, 5) -- (19, 5);
    
        \end{tikzpicture}
        \caption{original $\DT$}     
    \end{subfigure}%
    \hfill
    \begin{subfigure}[h]{0.15\textwidth}
        \centering
        \begin{tikzpicture}[scale = 0.3, treenode/.style={circle, draw, thick}, leaf/.style={rectangle, draw, thick}]
            
        \node[treenode] at (10, 10) (x1){$x_1$};
        \node[treenode] at (6, 6) (x3){$x_3$};
        \node[treenode] at (14, 6) (x5){$x_5$};

        \node[leaf] at (4, 3) (leaf1){0};
        \node[treenode] at (8, 3) (x4){$x_4$};

        \node[treenode] at (12, 3) (x2){$x_2$};
        \node[leaf] at (16, 3) (leaf6){0};

        \node[leaf] at (7, 0) (leaf2){0};

        \node[leaf] at (11, 0) (leaf4){0};
        \node[leaf] at (13, 0) (leaf5){1};
    
        \draw (x1) -- (newx5);
        \draw (x3) -- (leaf1);
        \draw (x3) -- (x4);
        \draw (x4) -- (leaf2);
        \draw (x1) -- (x5);
        \draw (x5) -- (x2);
        \draw (x2) -- (leaf4);
        \draw (x2) -- (leaf5);
        \draw (x5) -- (leaf6);

        \draw[->, very thick] (17, 5) -- (19, 5);
    
        \end{tikzpicture}
        \caption{$x_4 \gets 0$}     
    \end{subfigure}%
    \hfill
    \fbox{%
    \begin{subfigure}[h]{0.6\textwidth}
        \centering
        \begin{subfigure}[h]{0.4\textwidth}
            \centering
            \begin{tikzpicture}[scale = 0.3, treenode/.style={circle, draw, thick}, leaf/.style={rectangle, draw, thick}]
            
            \node[treenode] at (10, 10) (x1){$x_1$};
            \node[leaf] at (7, 6) (collapsed0) {0};
            \node[treenode] at (13, 6) (x5){$x_5$};
            \node[treenode] at (11, 3) (x2){$x_2$};
            \node[leaf] at (10, 0) (leaf4){0};
            \node[leaf] at (12, 0) (leaf5){1};
            \node[leaf] at (15, 3) (leaf6){0};
            
            \draw (x1) -- (collapsed0);
            \draw (x1) -- (x5);
            \draw (x5) -- (x2);
            \draw (x2) -- (leaf4);
            \draw (x2) -- (leaf5);
            \draw (x5) -- (leaf6);
            
            \draw[->, very thick] (18, 5) -- (20, 5);
            \end{tikzpicture}
        \end{subfigure}%
        \hfill %
        \begin{subfigure}[h]{0.5\textwidth}
            \centering
            \begin{tikzpicture}[scale = 0.3, treenode/.style={circle, draw, thick}, leaf/.style={rectangle, draw, thick}]
                
            \node[treenode] at (10, 10) (x1){$x_1$};
            \node[treenode] at (6, 6) (newx5){$x_5$};
            \node[treenode] at (14, 6) (x5){$x_5$};
            \node[treenode] at (12, 3) (x2){$x_2$};
            \node[leaf] at (11, 0) (leaf4){0};
            \node[leaf] at (13, 0) (leaf5){1};
        
            \node[treenode] at (16, 3) (newx2){$x_2$};
            \node[leaf] at (15, 0) (leaf6){0};
            \node[leaf] at (17, 0) (leaf7){0};
    
            \node[treenode] at (4, 3) (newx2l){$x_2$};
            \node[treenode] at (8, 3) (newx2r){$x_2$};
            \node[leaf] at (3, 0) (leaf1){0};
            \node[leaf] at (5, 0) (leaf2){0};
            \node[leaf] at (7, 0) (leaf3){0};
            \node[leaf] at (9, 0) (leaf9){0};
        
            \draw (x1) -- (newx5);
            \draw (x1) -- (x5);
            \draw (x5) -- (x2);
            \draw (x5) -- (newx2);
            \draw (x2) -- (leaf4);
            \draw (x2) -- (leaf5);
        
            \draw (newx2) -- (leaf6);
            \draw (newx2) -- (leaf7);
        
            \draw[->, blue, thick] (leaf4) .. controls (12, -2) .. (leaf5); 
        
            \draw[->, red, thick] (leaf6) .. controls (13, -3) .. (leaf4);
            \draw[->, red, thick] (leaf7) .. controls (15, -3) .. (leaf5);
        
            \draw (newx5) -- (newx2l);
            \draw (newx5) -- (newx2r);
            \draw (newx2l) -- (leaf1);
            \draw (newx2l) -- (leaf2);
            \draw (newx2r) -- (leaf3);
            \draw (newx2r) -- (leaf9);
        
            \draw[->, olive, thick] (leaf1) .. controls (7, -4) .. (leaf4);
            \draw[->, olive, thick] (leaf2) .. controls (9, -4) .. (leaf5);
            \draw[->, olive, thick] (leaf3) .. controls (11, -4) .. (leaf6);
            \draw[->, olive, thick] (leaf9) .. controls (13, -4) .. (leaf7);
        
            \end{tikzpicture}     
        \end{subfigure}%
        \caption{0-balancing}
    \end{subfigure}}
    }
    \caption{Illustration of \cref{step3a,step3b,step3c} in $\CDT$ construction}
    \label{fig: entire process}
\end{figure}

Given any $s$-long bit-string $a = (a_1,a_2,\dots,a_s)$, we can walk along the CDT using $a$ as an "instruction set". In other words, we walk from the root to a node of the tree by using $a$ to make choices at the degree-2 nodes and otherwise following the degree one-edges. If this walks ends at a node $w$ (possibly leaf node) of the CDT, we denote the corresponding ordered restriction $\alpha_w$ by $\CDT^{(a)}(F,\rhotilde)$, else $\CDT^{(a)}(F,\rhotilde)$ is undefined. When this node is a leaf node, then it is labelled either 0 or 1. In this case, we further enhance this definition as follows.

\begin{definition}
    Let $F$ be a formula and $\rhotilde$ an associated restriction tree. For any bit-string $a = (a_1,\dots,a_s)$ and $z \in \zo$, define
    \[ 
        \CDT^{(a)}_z(F,\rhotilde) = \begin{cases} 
            \alpha_w & \text{ if the walk according to instruction set "$a$" ends on a leaf $w$ labelled $b$, }\\
            \bot & otherwise. \qedhere
        \end{cases}
    \]
\end{definition}

\subsection{Unpacking the CDT}

Fix a formula $F$ on $n$ variables and an associated restriction tree $\rhotilde\colon T_F \to \{0,1,*\}^n$. Let $a = (a_1,\dots,a_s)$ be an $s$-bit-string with $s \geq 1$. Let us assume $F = F_1 \vee F_2 \vee \dots \vee F_m$ is a disjunction. In this section, we try to understand when $\CDT^{(a)}_0(F,\rhotilde)$ exists.

\medskip

Suppose $\CDT^{(a)}_0(F,\rhotilde)$ exists and is the ordered restriction $\alpha$. Then, the following must be true.
\begin{itemize}
    \item There exists a unique $\ell \in [m]$ such that for all $\ell' < \ell$, we have $F_{\ell'}|_{\rhotilde(F)} \equiv 0$ and $F_{\ell}|_{\rhotilde(F)}\not\equiv \text{constant}$.
    \item Let $\Gamma$ be $\CDT(F_\ell, \rhotilde|_{F_\ell})$. Let $\Gamma'$ be the tree obtained by restricting $\Gamma$ by $\rhotilde(F)$ and $\Gamma''$ be the 0-balancing of $\Gamma'$. There must be some $1 \leq r \leq s$ such that, a walk to a leaf of $\Gamma''$ using instruction set $a_{\leq r}$ leads us to a leaf in $\Gamma''$. Let $\alpha'$ be the corresponding ordered restriction (which ought to be a prefix of $\alpha$). 
    \item $\CDT^{(a_{>r})}_0(F|_{\alpha'}, \rhotilde)$ exists, and is $\alpha''$ say. In that case, $\alpha = (\alpha', \alpha'')$. 
\end{itemize}

Let us peer deeper into the balancing operation. Since $\Gamma''$ is the $0$-balancing of $\Gamma'$, we must have that $\ass(\alpha') = \beta$ for some $1$-leaf $\beta$ of $\Gamma''$ and this $\beta$ is a $1$-leaf of $\Gamma = \CDT(F_\ell, \rhotilde|_{F_{\ell}})$ as well. In fact, $\beta$ must be consistent with $\rhotilde(F)$ as this path in $\CDT(F_\ell, \rhotilde|_{F_\ell}) = \Gamma$ survived in $\Gamma'$ as well. Thus, there is some instruction set $b \in \set{0,1}^t$, for some $t \geq r$, such that $\CDT^{(b)}_1(F_\ell, \rhotilde|_{F_\ell}) = \beta$. 

Let us focus on the differences between the ordered restriction $\alpha'$ and $\beta$. We know there were $t$ degree-2 nodes on the path to $\beta$ in $\CDT(F_\ell, \rhotildel)$, and there were $r$ degree-2 nodes on the path to $\alpha'$ in $\Gamma''$. Thus, among the $t$ degree-2 nodes on the path to $\beta$, we must have that $t-r$ of them belong to $\dom(\rhotilde(F))$ (with $\beta$ being consistent with $\rhotilde(F)$) and the path to $\alpha'$ uses $a_{\leq r}$ as instructions for the other $r$ nodes (instead of whatever route was taken by the path to $\beta$). 

\medskip

We summarize this discussion in the following lemma. We will be using this lemma for a \emph{random} restriction tree $\brhotilde$ (chosen according to a suitable distribution). To distinguish the quantities that depend on this random variable from the rest, we use bold font to indicate all the quantities (including $\brhotilde$ itself) that are functions of $\brhotilde$. 

\begin{lemma}[Unpacking $\CDT^{(a)}_0(F,\brhotilde)$]
\label{lem:unpacking}
    Let $F= F_1 \vee \dots \vee F_m$ be a formula and $\brhotilde\colon T_F \to \bra{0,1,*}^n$ an associated restriction tree. Let $s \geq 1$ and $a \in \zo^s$. 
    
    Then $\CDT^{(a)}_0(F,\brhotilde)$ exists if and only if there exist 
    \begin{itemize}
        \item $\ell \in [m]$
        \item non-negative integer $r \in [s]$,
        \item non-negative integer $t \geq r$
        \item a bit-string $b \in \zo^t$ and 
        \item $Q \in \binom{[t]}{r}$
    \end{itemize} 
    such that the following three conditions $\calA, \calB, \calC$ are met. 
    \begin{description}
        \item [$\calA(\ell, t, b)$:]
        \renewcommand{\theenumi}{(\roman{enumi})}
        \begin{enumerate}
            \item \label{item:Ai} $F_{\ell'}|_{\brhotilde(F)} \equiv 0$ for all $\ell' < \ell$,
            \item \label{item:Aii}$\CDT_1^{(b)}(F_\ell, \brhotilde|_{{F_\ell}})$ exists (and is $\bbeta$ say).
            \item \label{item:Aiii} $\bbeta$ is consistent with $\brhotilde(F)$,
        \end{enumerate}
        \item[$\calB(\ell,t,b,r,Q,a_{\leq r})$:] 
        \begin{enumerate}
            \item \label{item:Bi}$Q$ identifies $\stars(\brhotilde(F))$ within $\dom(\bbeta) \cap \stars(\brhotilde(F_\ell))$.
            \item  \label{item:Bii}Let $\balphap$ is the ordered restriction obtained by modifying $\bbeta$ by replacing the assignment of the $r$ variables in $\dom(\bbeta) \cap \stars(\brhotilde(F_\ell))$ identified by $Q$ by $a_{\leq r}$. We denote this process as "$\balphap \xleftarrow[Q\gets a_{\leq r}]{\stars{\brhotilde(F_\ell)}} \bbeta$". Then $\balphap \in \ZASS(F_\ell,\brhotilde(F),\bbeta)$
        \end{enumerate}
        \item[$\calC(\ell,t,b,r,Q,a)$:] $\CDT^{(a_{>r})}_0(F|_{\balphap},\brhotilde)$ exists (is $\balphapp$ say).
    \end{description}
    Furthermore, when $\CDT^{(a)}_0(F,\brhotilde)$ exists, we have $\CDT^{(a)}_0(F,\brhotilde) = (\balphap,\balphapp)$.
\end{lemma}

When clear from context, we will drop the arguments $\ell, t, b, r, Q, a$ from the properties $\calA, \calB$ and $\calC$. 

\section{Downward closure property}\label{sec:downward}

Let $\rho , \rho'$ be two restrictions on the variable set $V$ and let $T \subseteq V$ any subset of the variables. Recall that we say that $\rho' \cleq_T \rho$ if (1) $\stars(\rho') \cap T \subseteq \stars(\rho) \cap T$ and (2) $\rho_1$ and $\rho_2$ are consistent. We say that  a set $\calF$ of restrictions is downward-closed with respect to the set of variables $T$ if the following holds for any pair of restrictions $\rho, \rho'$
\[\rho \in \calF \text{ and } \rho' \cleq_T \rho \implies \rho' \in \calF.\]

\noindent
We now extend this definition of downward-closed sets to restriction trees.

\begin{definition}[downward-closed set of restriction trees]
\label{def:downward closet set of trees of restrictions}
Let $F$ be a formula on the variable set $V$ and $\rhotilde,\rhotilde'\colon T_F \to \bra{0,1,*}^V$ be two associated restriction trees. Let $S \subseteq V$. We say $\rhotilde' \cleq_{S} \rhotilde$ iff for all $G \in T_F$, we have $\rhotilde'(G) \cleq_S \rhotilde(G)$.

We call a set $\calT$ of restriction trees \emph{downward-closed with respect to the variable set $T$} if the following holds for any pair of restriction trees
\[\rhotilde \in \calT \text{ and } \rhotilde' \cleq_{T} \rhotilde \implies \rhotilde' \in \calT.\]

If $T = V$ (the full set of variables), then we drop the subscript $T$ in the above definitions.
\end{definition}
It is evident that if $\calT$ and $\calT'$ are two downward-closed set of restriction trees with respect to a variable set, so is their intersection. The key property that enables our proof of the main lemma is the following downward-closure property.

\begin{lemma}\label{lem:downward closure lemma}
Let $F = F_1 \vee F_2 \vee \dots \vee F_m$ be a formula on variable set $V$ and $\rhotilde\colon T_F \to \bra{0,1,*}^{|V|}$ be an associated restriction tree. Let $s \in \Z_{> 0}$, $a \in \zo^s$ and $\alpha$ be an ordered restriction such that 
    \[
        \CDT^{(a)}_0(F,\rhotilde) = \alpha.
    \] 
Suppose $\rhotilde'\colon T_F \to \bra{0,1,*}^{|V|}$ is another restriction tree satisfying
\begin{itemize}\itemsep 0pt
\item $\rhotilde' \cleq \rhotilde$ and 
\item $\rhotilde'(G)|_{\dom(\alpha)} = \rhotilde(G)|_{\dom(\alpha)}$ for all $ G \in T_F$,
\end{itemize} 
then $\CDT^{(a)}_0(F,\rhotilde') = \alpha$.

Similarly, when $F = F_1 \wedge F_2 \wedge \dots \wedge F_m$, the same holds for ``$\CDT^{(a)}_1(F,\rhotilde) = \alpha$''.
\end{lemma}

\noindent
Note that the lemma implies that the set $\calT_{F,a,\alpha}:=\{\rhotilde \colon \CDT^{(a)}_0(F,\rhotilde) = \alpha\}$ is downward-closed with respect to the variable set $V\setminus\dom(\alpha)$.

\begin{proof}
We are given that $\rhotilde'$ and $\rhotilde'$ behave identically on $\dom(\alpha)$, and $\rhotilde'$ only sets more variables (all of them outside of $\dom(\alpha))$ than $\rhotilde$. The proof is by induction on the depth and number of variables in the formula. 

\paragraph{{Base case:}}
The base case is when $F|_{\rhotilde(F)}$ is a literal or a constant. The lemma is clearly true in this case as $\rhotilde'$ only sets more variables than $\rhotilde$ and does not change the variables in $\dom(\alpha)$.
\paragraph{{Induction step:}} 
Let $F$ be a formula of depth $d$ on the variable set $[n]$. Assume the lemma is true for all formulae of either depth less than $d$ or involving less than $n$ variables.

By the Unpacking lemma~(\cref{lem:unpacking}), we have that $\CDT^{(a)}_0(F,\rhotilde) = \alpha$ if and only if there exist $\ell, r, t, b, Q$ and ordered restrictions $\alpha', \alpha'', \beta$ such that the following are true.
\renewcommand{\theenumi}{(\roman{enumi})}%
\begin{enumerate}
    \item \label{i}$F_{\ell'}|_{\rhotilde(F)} \equiv 0$ for all $\ell' < \ell$,
    \item \label{ii}$\CDT_1^{(b)}(F_\ell, \rhotilde|_{{F_\ell}})=\beta$,
    \item \label{iii}$\beta$ is consistent with $\rhotilde(F)$,
    \item \label{iv}$Q$ identifies $\stars(\rhotilde(F))$ within $\dom(\beta) \cap \stars(\rhotilde(F_\ell))$.
    \item \label{v}$\alpha' \in \ZASS(F_\ell,\rhotilde(F),\beta)$ where $\alpha' \xleftarrow[Q\gets a_{\leq r}]{\stars{\rhotilde(F_\ell)}} \beta$ (i.e., $\alpha'$ is the ordered restriction obtained by modifying $\beta$ by replacing the assignment of the $r$ variables in $\dom(\bbeta) \cap \stars(\brhotilde(F_\ell))$ identified by $Q$ by $a_{\leq r}$). Note that $\alpha' \in \ZASS(F_\ell,\rhotilde(F),\beta)$ ensures that $\alpha'$ is a 0-path in the decision tree $\Gamma''$ where $\Gamma''$ is defined as follows:
    \[
        \Gamma = \CDT(F_\ell,\rhotilde_{F_\ell}) \quad\stackrel{\text{Apply $\rhotilde(F)$}}{\leadsto}\quad \Gamma' \quad\stackrel{\text{0-balance}}{\leadsto}\quad \Gamma''.
    \]
    
    \item \label{vi}$\CDT^{(a_{>r})}_0(F|_{\alpha'},\rhotilde) = \alpha''$.
    \item \label{vii}$\alpha = (\alpha',\alpha'')$.
\end{enumerate}

We will demonstrate that for the same $\ell, r, t, b, Q$ and ordered restrictions $\alpha',\alpha'', \beta$ all the above conditions continue to hold good when $\rhotilde$ is replaced by $\rhotilde'$. This will prove that $\CDT^{(a)}_0(F,\rhotilde')= \alpha$. 

\cref{vii} is trivially true as this is independent of $\rhotilde$ or $\rhotilde'$. The other conditions are met for the following reasons. We first observe that since $\alpha' \in \ZASS(F_\ell,\rhotilde(F),\beta)$, we have $\dom(\beta) = \dom(\alpha') \subseteq \dom(\alpha)$.
\begin{itemize}
    \item \cref{i,iii} continue to hold good when when more variables are set from $\rhotilde$ to $\rhotilde'$.
    \item \cref{ii,vi} are true when $\rhotilde$ is replaced by $\rhotilde'$ due to the inductive assumption (since $F_\ell$ is a formula of smaller depth, $F|_{\alpha'}$ is a formula on fewer variables and $\rhotilde'$ does not alter the variables in $\dom(\beta)=\dom(\alpha')$ or $\dom(\alpha'')$).
    \item Since the variables in $\dom(\beta)=\dom(\alpha')\subseteq \dom(\alpha)$ are unaltered by $\rhotilde'$, we have 
    \begin{align*}
        \stars(\rhotilde(F))\cap \dom(\beta)&=\stars(\rhotilde'(F))\cap \dom(\beta),\\ \stars(\rhotilde(F_\ell))\cap \dom(\beta)&=\stars(\rhotilde'(F_\ell))\cap \dom(\beta).
    \end{align*}
    Hence, if $Q$ identifies $\stars(\rhotilde(F))$ within $\dom(\beta) \cap \stars(\rhotilde(F_\ell))$, it also identifies $\stars(\rhotilde'(F))$ within $\dom(\beta) \cap \stars(\rhotilde'(F_\ell))$. Thus, \cref{iv} holds.
    
    \item As for \cref{v}, since $\dom(\beta)\cap \stars(\rhotilde(F_\ell))=\dom(\beta)\cap \stars(\rhotilde'(F_\ell))$ and $\alpha' \xleftarrow[Q\gets a_{\leq r}]{\stars{\rhotilde(F_\ell)}} \beta$, we also have  $\alpha' \xleftarrow[Q\gets a_{\leq r}]{\stars{\rhotilde'(F_\ell)}} \beta$. It is now easy to verify from the definition of $\ZASS$ (\cref{def:pivot}), if $\alpha' \in \ZASS(F_\ell,\rhotilde(F),\beta)$, then we also have $\alpha' \in \ZASS(F_\ell,\rhotilde'(F),\beta)$ since we are only setting more variables. Thus \cref{v} also holds.
        
\end{itemize}
Thus, we have proved the claim.
\end{proof}

All of the above works even when dealing with the representation of restrictions given by pairs $(\sigma,S)$ (see \cref{sec:Sampleable}). In this case, the notion of downward closure is the standard definition of downward closure of sets. \cref{lem:downward closure lemma} merely re-stated in this language is the following

\begin{lemma}\label{lem:downward-redundant}
    Let $F = F_1 \vee F_2 \vee \dots \vee F_m$ be a formula on variable set $V$ and  $(\sigma, \Stilde)$ be a representation of associated restriction tree $\rhotilde\colon T_F \to \bra{0,1,*}^{|V|}$ (i.e, $\rho = \rho_{(\sigma,\Stilde)}$). Let $s \in \Z_{> 0}$, $a \in \zo^s$ and $\alpha$ be an ordered restriction such that 
        \[
            \CDT^{(a)}_0(F,\rhotilde) = \alpha.
        \] 
    Suppose $(\sigma,\Stilde')$ is a representation of another restriction tree $\rhotilde'\colon T_F \to \bra{0,1,*}^{|V|}$ satisfying
    \begin{itemize}\itemsep 0pt
    \item $\Stilde'(G) \subseteq \Stilde(G)$ for all $G \in T_F$ and 
    \item $\Stilde'(G) \cap {\dom(\alpha)} = \Stilde(G)\cap {\dom(\alpha)}$ for all $ G \in T_F$,
    \end{itemize} 
    then $\CDT^{(a)}_0(F,\rhotilde') = \alpha$.
    
    Similarly, when $F = F_1 \wedge F_2 \wedge \dots \wedge F_m$, the same holds for ``$\CDT^{(a)}_1(F,\rhotilde) = \alpha$''.
\end{lemma}
We will complete this discussion by extending the definition of downward-close to this representation of restrictions. 
\begin{definition}\label{def:rep downward}
    Let $T$ be any subset of the variable set $V$.
    For any pair of sets $S, S'\subseteq V$, we say that $S' \subseteq_T S$ if $S' \cap T \subseteq S \cap T$ and $S'\setminus T = S \setminus T$. Similarly, for any pair $\Stilde, \Stilde' \colon T_F \to 2^V$, we say that $\Stilde' \subseteq_T \Stilde$ if for $\forall G \in T_F, \Stilde'(G) \subseteq_T \Stilde(G)$.

    A set of restrictions $\calF \subseteq \zo^V \times 2^V$ (given by their representations) is downward closed with respect to variable set $T$ if the following holds for every pair of representations $(\sigma,S)$ and $(\sigma', S')$:
    \[(\sigma,S) \in \calF \text{ and } S' \subseteq_T S \text{ and } \sigma|_{\overline{S}} \equiv \sigma'|_{\overline{S}}\implies (\sigma', S' ) \in \calF.\]
    
    Similarly, a set of restriction trees $\calT$ (given by their representations) is downward closed with respect to the variable set $T$ if the following holds for any pair of restriction trees $(\sigma, \Stilde)$ and $(\sigma', \Stilde')$ 
    \[(\sigma,\Stilde) \in \calT \text{ and } \Stilde' \subseteq_T \Stilde \text{ and } \sigma|_{\overline{\Stilde(F)}} \equiv \sigma'|_{\overline{\Stilde(F)}} \implies (\sigma', \Stilde' ) \in \calT.\]
\end{definition}
Thus, \cref{lem:downward-redundant} implies that the set $\calT_{F,a,\alpha}:=\{(\sigma, \Stilde) \colon \CDT^{(a)}_0(F,\rhotilde_{(\sigma,\Stilde)}) = \alpha\}$ is downward-closed with respect to the variable set $V\setminus\dom(\alpha)$.

\section{Bounds on criticality}\label{sec:criticality}
In this section, we prove \cref{thm:AC0criticality} (the criticality result for $\AC^0$ formulae). To this end, we first define $\lambda(F)$, the bound on criticality that we eventually prove. We then define a sampling procedure to sample random restriction trees $\brhotilde$ for a given formula $F$ such that the marginal distribution $\brhotilde(F)$ (i.e, the distribution of the restriction corresponding to the entire formula) is the standard $p$-random restriction. Finally, we state and prove the main inductive lemma (\cref{lem:main}) that proves \cref{thm:AC0criticality}.  

We begin by defining $\lambda(F)$ for any $\AC^0$-formula.
\begin{definition}[lambda]
    \label{def:criticality}
    For a positive integer $S \in \Z_{>0}$ and non-negative integer $d \in \Z_{\geq 0}$, define 
    \begin{align*}\lambda_{S,d} &: = 32^{d+1}\left(\frac{\log S}{d}+1\right)^d = 32^{d+1}\left(\frac{\log (2^d \cdot S)}{d}\right)^d.
    \end{align*}
    Given an $\AC^0$ formula $F$ of depth $d+1$ and size $S$, define $\lambda(F) := \lambda_{S,d+1}$.
\end{definition}
Note, that the above expression simplifies to 32 for depth-1 formulae (i.e., terms and clauses), where we have used the convention that $\frac00=1$.

\begin{claim}\label{clm:lambdainc} $8\lambda_{S,d} \leq \lambda_{S,d+1}$
\end{claim}
\begin{proof}
$8\lambda_{S,d}
\leq 8\cdot32^{d+1}\left(\frac{\log 2^{d+1}S}{d}\right)^d 
= \frac{\lambda_{S,d+1}}{4\cdot \log 2^{d+1}S}\frac{(d+1)^{d+1}}{d^d}
\leq  \frac{\lambda_{S,d+1}}4 \frac{e(d+1)}{d+1+\log S}\leq \lambda_{S,d+1}
.$
\end{proof}    


\subsection{Sampling restriction trees}

We begin by recalling the definition of the classical $p$-biased distribution and the $p$-random restriction $\calR_p$ distribution over restrictions.

\begin{definition}[$p$-biased distribution]
    For $p \in [0,1]$ and variable set $V$, the $p$-biased distribution $\mu_p(V)$ is the distribution on the power set $2^V$ where a set $\bS \in 2^V$ is sampled as follows: 
    
    For each $v \in V$, independently set "$v \in \bS$" with probability $p$.  
    
    We will express this succinctly as "$\bS \gets_p 2^V$". 
\end{definition}

\begin{definition}[$p$-random restriction]\label{def:Rp}
For $p \in [0,1]$ and a variable set $V$, $\calR_p([n])$ is the distribution on representations of restrictions obtained by independently sampling a uniformly random string $\bsigma \gets \zo^V$ and a set $\bS \gets_p 2^V$ and outputting the pair $(\bsigma, \bS)$. The corresponding random restriction $\brho$ is given by $\brho \gets \rho_{(\bsigma,\bS)}$.

\end{definition}

We now extend this definition to distribution over restriction trees. Given a formula $F$, we say that $\ptilde\colon T_F \to [0,1]$ is a valid set of probabilities if whenever $G$ is a sub-formula of $H$, we have $\ptilde(G) \geq \ptilde(H)$. 

\begin{definition}[$\calRtilde_{p}$-distribution]\label{def:random restriction tree}
Let $F$ be a formula on the variable set $V$ and $\ptilde \colon T_F \to [0,1]$ be a valid set of probabilities. The distribution $\calRtilde_{\ptilde}(F)$ on representations of restriction trees is the the one obtained from the following sampling algorithm.
    \begin{enumerate}
        \item Choose a uniformly random string $\bsigma \gets \zo^V$.
        \item For each $G \in T_F$, choose independently a random $\bS_G \gets_{q_G} 2^V$ where 
        \[q_G := \frac{\ptilde(G)-\ptilde(\text{parent}(G))}{1-\ptilde(\text{parent}(G))}. \] (Here, we follow the convention that $\ptilde(\text{parent}(F))=0$).

        Note $v \notin \bS_G$ with probability $\nicefrac{(1-\ptilde(G))}{(1-\ptilde(\text{parent}(G)))}$.
        \item For each $G \in T_F$, let $G_0 := G, G_1, \dots, G_k := F$ be the sequence of formulae from $G$ to the root $F$ in the formula tree $T_F$. Set $\bStilde(G) \gets \bS_{G_0} \cup \bS_{G_1}\cup  \cdots \cup \bS_{G_k}$.
        \item Output the pair $(\bsigma, \bStilde)$.
    \end{enumerate}
    The corresponding random restriction tree $\brhotilde$ is given by $\brhotilde \gets \rhotilde_{(\bsigma,\bStilde)}$. 

For any $p \in [0,\nicefrac1{\lambda(F)}]$, let $\dist$ denote the distribution $\calRtilde_{\ptilde}(F)$ where $\ptilde$ is defined as follows $\ptilde(F) = p$ and for all $G \in T_F$ other than $F$, we have $\ptilde(G) = \nicefrac1{8\lambda(G)}.$\footnote{For this to be well-defined, we need $\lambda(F) \geq 8\lambda(G)$ for any sub-formula $G$ of $F$. This follows from \cref{clm:lambdainc}}
\end{definition}

It follows from the definition of $\calRtilde_{\ptilde}(F)$ that the marginal distribution $\brhotilde(G)$ on any sub-formula $G \in T_F$ is distributed exactly according to the distribution $\calR_{\ptilde(G)}$. 

\subsection{Main Lemma}

We are now ready to state the main lemma of the paper.

\begin{lemma}\label{lem:main}
Let $d \geq 0$ and $F=F_1 \vee F_2 \vee \dots \vee F_m$ be an $\AC^0$ formula of size $S$ and depth $d+1$ on $n$ variables. Let $\calT$ be any set of downward-closed set of representations of restriction trees with respect to the variables of the formula $F$, then for all integers $s \geq 1$ and $a \in \zo^s$, 
\[ \prob{(\bsigma,\bStilde) \sim \dist}{ \CDT^{(a)}_0(F,\brhotilde_{(\bsigma,\bStilde)}) \text{ exists} \mid (\bsigma,\bStilde) \in \calT} \leq (p \cdot \lambda(F))^s.\]
The statement for conjunctions $F=F_1 \wedge F_2 \wedge \dots \wedge F_m$ is identical with $\CDT^{(a)}_0$ replaced by $\CDT^{(a)}_1$.    
\end{lemma}

\cref{thm:AC0criticality} stated in the introduction clearly follows from the above lemma. The above lemma is stronger than what is needed for \cref{thm:AC0criticality} as it proves the statement even when conditioned under any downward-closed set of restriction trees. This stronger statement is needed for the inductive proof to go through.

\begin{proof}
The proof is by induction on the depth $d$ of the formula and the number of variables in the formula $F$.

Let us begin with the base case (depth-1 $\AC^0$-formulae). The proof of this is similar to the proof of \cite[Lemma 3.4]{Hastad2014}. The proof of the base case is written in a slightly more complicated fashion than it needs to be as it then serves as a warmup to one of the key claims (\cref{clm:hastad}) in the proof of the induction step (the base ). 

\paragraph{Base case:} The base case is when $F$ is a depth-1 formula and we need to bound the probability by $(32p)^s$ since in this case $\lambda(F)=32$. A depth-1 formula is a term or a clause. Without loss of generality let's assume that $F$ is a clause of the form $x_1 \vee \dots \vee x_m$, where the $x_i$'s are distinct variables. 

For a given $a \in \zo^s$, let $(\sigma,\Stilde)$ be such that "$\CDT^{(a)}_0(F,\rhotilde_{(\sigma,\Stilde)}) \text{ exists}$" and $(\sigma,\Stilde)\in \calT$. Then there exists a unique subset of variables $Q^{(a)}_{(\sigma,\Stilde)}\subset [m]$ of size $s$  such that $\Stilde(F)\cap [m]=Q^{(a)}_{(\sigma,\Stilde)}$ and for all variables $i \in [m] \setminus Q^{(a)}_{(\sigma,\Stilde)}$, we have $\sigma(x_i) =0$. For any $b \in \zo^s$, define $\sigma^{(b)}$ to be the global assignment that agrees with $\sigma$ outside $Q^{(a)}_{(\sigma,\Stilde)}$ and is equal to $b$ within $Q^{(a)}_{(\sigma,\Stilde)}$. 
Since $Q^{(a)}_{(\sigma,\Stilde)} \subseteq \Stilde(F)$ and $\calT$ is downward-closed, we have that these $2^s$ different representations $(\sigma^{(b)},\Stilde)$ are also in $\calT$. Furthermore, since $\rhotilde_{(\sigma,\Stilde)} = \rhotilde_{(\sigma^{(b)},\Stilde)}$, we have that all these $2^s$ representations also satisfy "$\CDT^{(a)}_0(F,\rhotilde_{(\sigma^{(b)},\Stilde)}) \text{ exists}$". We can hence conclude that
\begin{align}
    &\prob{(\bsigma,\bStilde) \sim \dist}{ \CDT^{(a)}_0(F,\brhotilde_{(\bsigma,\bStilde)}) \text{ exists} \mid (\bsigma,\bStilde) \in \calT} \nonumber\\
    &\leq 2^s \cdot \prob{(\bsigma,\bStilde) \sim \dist}{ \CDT^{(a)}_0(F,\brhotilde_{(\bsigma,\bStilde)}) \text{ exists} \text{ and } \forall i \in Q^{(a)}_{(\bsigma,\bStilde)}, \bsigma(x_i)=1 \mid (\bsigma,\bStilde) \in \calT}.\label{eq:base2powers}
\end{align}

Let $\calE_a = \bra{(\sigma,\Stilde )\colon \CDT^{(a)}_0(F,\rhotilde_{(\sigma,\Stilde)}) \text{ exists} \text{ and } \forall i \in Q^{(a)}_{(\bsigma,\bStilde)}, \bsigma(x_i)=1 }$. We need to bound the quantity $\mu(\calE_a \cap \calT)/\mu(\calT)$. For every $(\sigma,\Stilde) \in \calE_a\cap \calT$, define the set $N(\sigma,\Stilde)$ as outlined below. 
\begin{equation} N(\sigma, \Stilde) := \left\{(\sigma,\Stilde') \colon \Stilde' \subseteq_{Q^{(a)}_{(\bsigma,\bStilde)}} \Stilde \right\}.\end{equation}
(Recall the definition of the notation "$\Stilde' \subseteq_T \Stilde$" from \cref{def:rep downward}).

It follows from the definition of $N(\sigma,\Stilde)$ and the distribution $\dist$, that
\begin{equation}\mu(\sigma,\Stilde) = p^s \cdot \mu(N(\sigma,\Stilde)).
\end{equation}
We now make the following observations about $N(\sigma,\Stilde)$.
\begin{itemize}
    \item Since $\calT$ is downward closed and $(\sigma,\Stilde) \in \calT$, we have $N(\sigma,\Stilde) \subseteq \calT$.
    \item Exactly one element of $N(\sigma,\Stilde)$, namely $(\sigma,\Stilde)$, satisfies $\calE_a$.
    \item For distinct $(\sigma,\Stilde)$, the corresponding $N(\sigma,\Stilde)$ are disjoint.
\end{itemize}
We can now bound $\prob{}{\calE_a \mid \calT}$ using the above observations as follows:
\begin{align*}
    \prob{}{\calE_a \mid \calT} &= \frac{\mu(\calE_a \cap \calT)}{\mu(\calT)} = \frac{\sum_{(\sigma,\Stilde) \in \calE_a\cap \calT}\mu(\sigma,\Stilde)}{\sum_{(\sigma,\Stilde)\in \calE_a \cap\calT}\mu(N(\sigma,\Stilde))+\mu\left(\calT\setminus \bigcup_{(\sigma,\Stilde) \in \calE_a \cap \calT}N(\sigma,\Stilde)\right)}\\
    & \leq \frac{\sum_{(\sigma,\Stilde)\in \calE_a \cap \calT}\mu(\sigma,\Stilde)}{\sum_{(\sigma,\Stilde)\in \calE_a \cap \calT}\mu(N(\sigma,\Stilde))}
     = p^s.
\end{align*}
Combining the above bound with \eqref{eq:base2powers}, we thus have
\[\prob{(\bsigma,\bStilde) \sim \dist}{ \CDT^{(a)}_0(F,\brhotilde_{(\bsigma,\bStilde)}) \text{ exists} \mid (\bsigma,\bStilde) \in \calT} \leq (2p)^s,\]
concluding the base case of the induction.

 

\textbf{Induction step:} Let us assume without loss of generality that $F = F_1 \vee \dots \vee F_m$ and the main lemma holds for all formulae of smaller depth (in particular the $F_i$'s) and all formulae with smaller number of variables (in particular $F|_\beta$ for any non-trivial restriction $\beta$). By the Unpacking Lemma (\cref{lem:unpacking}) and a union bound we have that
\begin{equation}\label{eq:union1}\Pr_{(\bsigma,\bStilde)}\left[\CDT^{(a)}_0(F,\brhotilde_{(\bsigma,\bStilde)}) \text{ exists} \mid (\bsigma,\bStilde) \in \calT\right] \leq \sum_{r \in [s]} \sum_{\ell \in [m]} \sum_{t \colon t \geq r} \sum_{Q \in \binom{[t]}{r}} \sum_{b\in \zo^t} \Pr_{(\bsigma,\bStilde)}\left[ \calA \cap \calB \cap \calC \mid \calT\right], 
\end{equation}    
where $\calA, \calB$ and $\calC$ are as defined in \cref{lem:unpacking}.

For each fixing of $\ell, r, t,b, Q$ and $a$, we will bound the summand $\Pr\left[ \calA \cap \calB \cap \calC \mid \calT\right]$ in the above expression. Consider a $(\sigma,\Stilde) \in \calA \cap \calB \cap \calC \cap \calT$. Let $\beta = \CDT^{(b)}_1(F_\ell,\rhotilde_{(\sigma,\Stilde)}|_{F_\ell})$ which is guaranteed to exist as $(\sigma,\Stilde) \in \calA$. By property $\calB$, we have that there exist exactly $r$ variables in $\dom(\beta)\cap \stars(\rhotilde_{(\sigma,\Stilde)}|_{F_\ell})= \dom(\beta)\cap \Stilde(F_\ell)$ which belong to $\stars(\rhotilde_{(\sigma,\Stilde)})= \Stilde(F)$. Let us refer to this set of variables as $Q_{(\sigma,\Stilde)}$. This part of the proof is similar to the base case. For any $b \in \zo^s$, define $\sigma^{(b)}$ to be the global assignment that agrees with $\sigma$ outside $Q_{(\sigma,\Stilde)}$ and is equal to $b$ within $Q_{(\sigma,\Stilde)}$. 
Since $Q_{(\sigma,\Stilde)} \subseteq \Stilde(F)$ and $\calT$ is downward-closed, we have that these $2^r$ different representations $(\sigma^{(b)},\Stilde)$ are also in $\calT$. Furthermore, since $\rhotilde_{(\sigma,\Stilde)} = \rhotilde_{(\sigma^{(b)},\Stilde)}$, we have that all these $2^r$ representations also satisfy $\calA \cap \calB \cap \calC$. We can hence conclude that
\begin{align}
    \Pr_{(\bsigma,\bStilde)}\left[ \calA \cap \calB \cap \calC \mid \calT\right] \leq 2^r \cdot\Pr_{(\bsigma,\bStilde)}\left[ \calA' \cap \calB \cap \calC \mid \calT\right], \label{eq:2powerr}
\end{align}
where $\calA'(\ell,t,b)$ is a modification of $\calA$ (with respect to \cref{item:Aiii}) as follows:
\begin{description}
    \item [$\calA'(\ell, t, b)$:]
    \renewcommand{\theenumi}{(\roman{enumi})}
    \begin{enumerate}
        \item \label{item:Api}$F_{\ell'}|_{\brhotilde(F)} \equiv 0$ for all $\ell' < \ell$,
        \item \label{item:Apii}$\CDT_1^{(b)}(F_\ell, \brhotilde|_{{F_\ell}})$ exists (and is $\bbeta$ say).
        \item \label{item:Apiii}$\bbeta$ is consistent with $\bsigma$,
    \end{enumerate}
\end{description}
We thus, have 
\begin{equation}\label{eq:union}\Pr\left[\CDT^{(a)}_0(F,\brhotilde) \text{ exists} \mid (\bsigma,\bStilde)\in \calT\right] \leq \sum_{r \in [s]} \sum_{\ell \in [m]} \sum_{t \colon t \geq r} \sum_{Q \in \binom{[t]}{r}} \sum_{b\in \zo^t} 2^r \cdot \Pr_{(\bsigma,\bStilde)}\left[ \calA' \cap \calB \cap \calC \mid \calT\right]. 
\end{equation} 
For each fixed choice $r, \ell, t, b, Q$, the summand in the above expression can be factorized as 
\begin{equation}\label{eq:summand}\prob{}{\calA' \mid \calT} \cdot \prob{}{\calB \mid \calA' \cap \calT} \cdot \prob{}{\calC \mid  \calA' \cap \calB \cap \calT}.
\end{equation}    
The following three claims bound each of the terms in the above product.
\begin{claim}\label{clm:calC}
For a fixed $a, \ell, r, t,b$ and $Q$, we have 
\[\prob{(\bsigma, \bStilde)\sim \dist}{\CDT^{(a_{>r})}_0(F|_{\balphap}, \brhotilde_{(\bsigma,\bStilde)}) \text{ exists}\mid \calA'(\ell,t,b) \cap \calB(\ell,t,b,r,Q, a_{\leq r})\cap \calT} \leq (p\cdot \lambda(F))^{s-r}.\]
\end{claim}
\begin{proof}
We first note that the formula being considered in the above expression, namely $F|_{\balphap}$, is itself random since the ordered restriction $\balphap$ is random. To deal with this, we prove the above bound for each fixing of $\balphap$. More precisely, we rewrite the above expression as follows (here we not only fix $\balphap$, but also $\bbeta$).    
\[ \E_{\alpha',\beta}\left[\underbrace{\prob{\brhotilde}{\CDT^{(a_{>r})}_0(F|_{\alpha'}, \brhotilde) \text{ exists} \mid \calA' \cap \calB \cap \calT \cap \calE_{\alpha',\beta}}}\right],\] 
where the expectation over $\alpha'$ and $\beta$ is over the appropriate marginal distribution and $\calE_{\alpha',\beta}$ is the set of representations of restriction trees $(\bsigma,\bStilde)$ that satisfy $\balphap =\alpha'$ and $\bbeta = \beta$. We will prove that for any fixing of $\balpha$ and $\bbeta$, the indicated quantity in the above expression is at most $(p \cdot \lambda(F))^{s-r}$, which would imply the claim.

Consider any fixing $(\alpha',\beta)$ of $(\balphap,\bbeta)$. We first observe that since $\alpha'$ is a non-trivial ordered restriction (which is true since $r \geq 1$), the variable set of the formula $F|_{\alpha'}$ is less than that of $F$ and hence we can apply the inductive assumption provided the set  $\calA' \cap \calB \cap \calT \cap \calE_{\alpha',\beta}$ is downward closed with respect to the variables of $F|_{\alpha'}$. Below, we verify that this is indeed the case.
\begin{description}
    \item[$\calA'(\ell,t,b) \cap \calE_{\alpha',\beta}$:] We will show that each of the 3 items  of $\calA'\cap \calE_{\alpha',\beta}$ are downward-closed.
    \renewcommand{\theenumi}{(\roman{enumi})}%
    \begin{enumerate}
     \item $\calA'$\ref{item:Api} and $\calA'$\ref{item:Apiii} are clearly downward-closed.
     \item $\calA'$\ref{item:Apii} $\cap$ $\calE_{\alpha',\beta}$ is the event that ``$\CDT^{(b)}_1(F_\ell, \brhotilde_{(\bsigma,\bStilde)}|_{F_\ell}) = \beta$''. This is downward-closed on the variable set $[n]\setminus\dom(\beta) = [n]\setminus \dom(\alpha')$ by \cref{lem:downward-redundant}.
    \end{enumerate}
    \item[$\calB(\ell, t,b,r,Q.a_{\leq r})\cap \calE_{\alpha',\beta}$:] 
        $\calB$\ref{item:Bi} and $\calB$\ref{item:Bii}: 
        Both these conditions continue to hold good as long as the variables in $\dom(\beta) = \dom(\alpha')$ are unaltered. Since we care only about downward-closure on the variable set $[n]\setminus \dom(\alpha')$, we are fine (note this is not necessarily downward-closed on the entire set of variables).
\end{description}
Combined with the fact that $\calT$ is downward-closed, we have $\calA' \cap \calB \cap \calT \cap \calE_{\alpha',\beta}$ is downward-closed on $[n] \setminus \dom(\alpha')$ and hence by the inductive assumption, we have the required bound.   
\end{proof}    

\begin{claim}\label{clm:hastad}
For fixed $a, \ell, r, t, b$ and $Q$, we have $\prob{}{\calB \mid \calA' \cap \calT} \leq (8 \cdot p \cdot \lambda(F_\ell))^r$.
\end{claim}
As indicated earlier, the proof of this claim is similar in spirit to the proof of the base case, which in turn is similar to the proof of \cite[Lemma 3.4]{Hastad2014}. Things are however considerably more involved here and one has to do a careful conditioning argument to obtain the bound. 

\begin{proof}
It suffices if for each fixed choice of $a, \ell, r,t, b$ and $Q$, we prove
\[ \prob{(\bsigma,\bStilde) \sim \dist}{Q \text{ identifies } \bStilde(F) \text{ within } \dom(\bbeta) \cap \bStilde(F_\ell) \mid \calA'(\ell,t,b) \cap \calT} \leq q^r,\]    
where $q := (p/(\nicefrac{1}{8\lambda(F_\ell)}))= (8\cdot p \cdot \lambda(F_\ell))$.

Consider any $(\sigma,\Stilde)$ that satisfies the three properties (1) "$Q \text{ identifies } \Stilde(F) \text{ within } \dom(\beta) \cap \Stilde(F_\ell)$", (2) $\calA'(\ell,t,b)$  and (3) $\calT$. As before let $Q_{(\sigma,\Stilde)}$ be the set of $r$ variables in $\Stilde(F_\ell)\cap \dom(\beta)$ which belong to $\Stilde(F)$.  For every such $(\sigma,\Stilde)$, we define the set $N(\sigma,\Stilde)$ of representations of restrictions trees as follows. 
\begin{equation} N(\sigma, \Stilde) := \left\{(\sigma,\Stilde') \colon \Stilde'(G) \subseteq_{Q_{(\bsigma,\bStilde)}} \Stilde(G) \text{ for every } G \in T_F \setminus T_{F_\ell} \text{ and } \Stilde'(H) = \Stilde(H) \text{ for every } H \in T_{F_\ell}\right\}.\end{equation}

It follows from the definition of $N(\sigma,\Stilde)$ and the distribution $\dist$, that
\begin{equation}\mu(\sigma,\Stilde) = q^r \cdot \mu(N(\sigma,\Stilde)).\end{equation}
We now make the following observations about $N(\sigma,\Stilde)$.
\begin{itemize}
    \item Exactly one element of $N(\sigma,\Stilde)$, namely $(\sigma,\Stilde)$, satisfies property (1).
    \item For distinct $(\sigma,\Stilde)$, the corresponding $N(\sigma,\Stilde)$ are disjoint.
    \item Since $\calT$ is downward closed and $(\sigma,\Stilde) \in \calA' \cap \calT$, we have $N(\sigma,\Stilde) \subseteq \calA' \cap \calT$.
\end{itemize}
Putting these facts together, we have the following bound on the probability that we wish to bound.
\begin{align*}
    &\prob{(\bsigma,\Stilde)}{Q \text{ identifies } \bStilde(F) \text{ within } \dom(\bbeta) \cap \bStilde(F_\ell) \mid \calA'(\ell,t,b) \cap \calT}\\
    &\leq \frac{\sum_{(\sigma,\Stilde)}\mu(\sigma,\Stilde)}{\sum_{(\sigma,\Stilde)}\mu(N(\sigma,\Stilde))}
     = q^r,
\end{align*}
where the summation (in both the numerator and denominator) in the second step above is over all $(\sigma,\Stilde)$ that satisfy all three properties. This completes the proof of the claim.
\end{proof}

\begin{claim}\label{clm:multb}
For fixed $a, \ell, t$ and $b$, we have
\[ \eta(\ell, t, b) := \prob{}{\calA'(\ell, t,b) \mid \calT} \leq \left(\frac18\right)^t. \]
\end{claim}  
\begin{proof}
\begin{align*}
    \eta(\ell,t,b) & =  \prob{(\bsigma,\bStilde)\sim \dist}{\calA'(\ell, t,b) \mid \calT}\leq \prob{(\bsigma,\bStilde)\sim \dist}{\CDT^{(b)}_1(F_\ell, \brhotilde_{(\bsigma,\bStilde)}|_{F_\ell}) \text{ exists} \mid \calT}\\ 
    & = \prob{(\bsigma,\bStilde_{\boldsymbol{\ell}})\sim \calRtilde_{\ptilde(F_\ell)}(F_\ell)}{\CDT^{(b)}_1(F_\ell, \brhotilde_{(\bsigma,\bStilde_{\boldsymbol{\ell}})}) \text{ exists} \mid \calT }\\ 
    & \leq\left(\ptilde(F_\ell)\cdot \lambda(F_\ell)\right)^{t} =\left(\frac{1}{8\lambda(F_\ell)}\cdot \lambda(F_\ell)\right)^{t} = \left(\frac18\right)^t. 
\end{align*}
The last inequality follows from the induction assumption since $F_\ell$ has depth strictly smaller than that of $F$.
\end{proof}    

Plugging the results of these claims back into the the expression in \eqref{eq:summand}, we have
\[ \eta(\ell, t, b)\cdot (16\cdot p \cdot \lambda(F_\ell))^{r} \cdot (p\cdot \lambda(F))^{s-r} \leq \left(\frac18\right)^t \cdot \left(8\cdot p \cdot \lambda(F_\ell)\right)^{r} \cdot (p\cdot \lambda(F))^{s-r}\]


We need to bound the sum of this expression when summed over all $r,\ell, t, b, Q$ as given by \eqref{eq:union}. However, even if we just over all possible $\ell$ the sum turns out to be prohibitively expensive. To keep this sum over $\ell$ (and also $r,t,b,Q$) under control, we further observe that the events $\calA'(\ell,t,b)$ are mutually disjoint over disjoint $\ell, t, b$. This lets us conclude the following claim.

\begin{claim}\label{clm:subdistribution}    
\[ \sum_{\ell, t, b} \eta(\ell, t,b) = \sum_{\ell,t,b} \prob{(\bsigma,\bStilde)\sim \dist}{\calA'(\ell,t,b) \mid \calT} \leq 1. \]
\end{claim}  
\begin{proof}
We observe that given a $(\bsigma,\bStilde)$, there is at most one 
$\ell$ such that $F_{\ell'}|_{\brhotilde(F)}\equiv 0$ and $F_\ell|_{\brhotilde(F)} \not\equiv 0$. Fix such an $\ell$ (if one exists). Given this $\ell$, there is exactly one root-to-leaf path in $\CDT(F_\ell,\brhotilde|_{F_\ell})$ that is consistent with $\bsigma$. Let this be $\bbeta$ (if one exists). Let $t := |\dom(\beta) \cap \stars(\brhotilde(F_\ell))|$ and $b \in \zo^t$ the assignment to the $t$ degree-2 variables along the ordered restriction $\bbeta$. Thus, $(\bsigma,\bStilde)$ uniquely determines $(\ell,t,b)$ such that $\calA'(\ell,t,b)$ hold. Hence, $ \eta(\ell, t,b)$ is a sub-distribution.
\end{proof}

We now have all the ingredients to bound the quantity of concern. The rest of the proof is a roller-coaster ride along the Jensen highway. We now bound the quantity in \eqref{eq:union} as follows:

\begin{align}
\Pr_{(\bsigma,\bStilde)}\left[\CDT^{(a)}_0(F,\brhotilde_{(\bsigma,\bStilde)}) \text{ exists } \mid (\bsigma,\bStilde) \in \calT\right] &\leq \sum_{r,\ell,t,b,Q} 2^r \cdot \eta(\ell, t, b)\cdot \left(8\cdot p \cdot \lambda(F_\ell)\right)^{r} \cdot (p\cdot \lambda(F))^{s-r}\nonumber\\
& \leq \sum_{r,\ell,t,b} \eta(\ell, t, b)\cdot \left(16\cdot p \cdot \lambda(F_\ell)\right)^{r} \cdot (p\cdot \lambda(F))^{s-r}\cdot t^r\nonumber\\
& = \left(p \cdot \lambda(F)\right)^s \cdot \sum_{r,\ell} \left[\left(\frac{16\cdot \lambda(F_\ell)}{\lambda(F)}\right)^r \cdot \nu(\ell) \cdot \underbrace{\sum_{t,b} \frac{\eta(\ell,t,b)}{\nu(\ell)} \cdot t^r}\right]\label{eq:tr}
\end{align}
where $\nu(\ell) := \sum_{t,b}\eta(\ell,t,b)$. We only sum over those $\ell$ that satisfy $\nu(\ell) >0$. Observe that $\sum_\ell \nu(\ell) = \sum_{\ell,t,b} \eta(\ell,t,b) \leq1$. We first bound the quantity indicated (using underbraces) in the above expression using Jensen's inequality and \cref{clm:multb} as follows.
\begin{subclaim}
\[\sum_{t,b} \frac{\eta(\ell,t,b)}{\nu(\ell)} \cdot t^r \leq  \left(\log\left(\frac1{\nu(\ell)}\right)\right)^r.\]
\end{subclaim}
\begin{proof}
Rewriting $t^r$ as $(\log 2^t)^r$, the lefthand side can be written as $\sum_{t,b} \frac{\eta(\ell,t,b)}{\nu(\ell)} \cdot (\log 2^t)^r$. Since $\sum_{t,b} \frac{\eta(\ell,t,b)}{\nu(\ell)} = 1$, we can apply Jensen's inequality to the concave function $x \mapsto (\log x)^r$ to obtain
\begin{align*}
\sum_{t,b} \frac{\eta(\ell,t,b)}{\nu(\ell)} \cdot \left[\log 2^t\right]^r  & \leq \left[ \log\left(\sum_{t,b}  \frac{\eta(\ell,t,b)}{\nu(\ell)} \cdot 2^t \right)\right]^r\\
&\leq \left[ \log\left(\sum_{t,b}  \frac{1}{\nu(\ell)\cdot 4^t}  \right)\right]^r & [\text{Since $\eta(\ell,t,b) \leq 8^{-t}$ from \cref{clm:multb}}]\\
&\leq \left[ \log\left(\frac{1}{\nu(\ell)}\sum_{t}  \frac{1}{2^t}  \right)\right]^r&[\text{Since there are at most $2^t$ $b$'s}]\\
& \leq \left[ \log \frac1{\nu(\ell)}\right]^r \tag*{\qedhere}
\end{align*}
\end{proof}    
Substituting this bound back into the expression \eqref{eq:tr} above, we obtain
\begin{align*}
\Pr\left[\CDT^{(a)}_0(F,\brhotilde) \text{ exists} \mid \brhotilde \in \calT\right] &\leq  \left(p \cdot \lambda(F)\right)^s \cdot \sum_{r} \left[\left(\frac{16}{\lambda(F)}\right)^r \cdot \sum_\ell \nu(\ell) \cdot \left(\underbrace{\lambda(F_\ell)\cdot \log\left(\frac1{\nu(\ell)}\right)}\right)^r\right]
\end{align*}
We now apply AM-GM inequality and the definition of $\lambda(F_\ell)$ to bound the indicated quantity.
\begin{subclaim} Let $S_\ell := \size(F_\ell)$. Then,
\[\lambda(F_\ell)\cdot \log\left(\frac1{\nu(\ell)}\right) \leq 32^{d}\left[\frac{\log\left(\nicefrac{2^{d-1} \cdot S_\ell}{\nu(\ell)}\right)}{d}\right]^{d}.\]   
\end{subclaim}
\begin{proof}
    If $d = 1$, then $S_\ell = 1$ (recall the `size' defined in \cref{def: AC0}), $\lambda(F_\ell) = 32$ (recall \cref{def:criticality}). Thus, both sides of the above claim simplify to $32 \cdot \log(\nicefrac{1}{\nu(\ell)})$. 

    For larger $d$, 
    \begin{align*}
    \lambda(F_\ell)\cdot\log\left(\frac1{\nu(\ell)}\right) &= 32^{d} \left(\frac{\log 2^{d-1}\cdot S_\ell}{d-1}\right)^{d-1}\cdot \log\left(\frac1{\nu(\ell)}\right)&[\text{Since } F_\ell \in \AC^0[S_\ell,d]]\\
    & \leq 32^{d}\left[\frac{\log (2^{d-1} \cdot S_\ell) +\log\left(\frac1{\nu(\ell)}\right)}{d}\right]^{d} &[\text{Applying AM-GM inequality}]\\
    & = 32^{d}\left[\frac{\log\left(\nicefrac{2^{d-1} \cdot S_\ell}{\nu(\ell)}\right)}{d}\right]^{d}. 
    \end{align*}
\end{proof}

\noindent
Plugging this bound back into our expression, we have
\begin{align*}
\Pr\left[\CDT^{(a)}_0(F,\brhotilde) \text{ exists} \mid \brhotilde \in \calT\right] &\leq  \left(p \cdot \lambda(F)\right)^s \cdot \sum_{r} \left[\left(\frac{16\cdot 32^{d}}{\lambda(F)}\right)^r \cdot \underbrace{\sum_\ell \nu(\ell) \cdot \left(\frac{\log\left(\frac{2^{d-1} \cdot S_\ell}{\nu(\ell)}\right)}{d}\right)^{dr}}\right]
\end{align*}
We bound the indicated quantity using yet another application of Jensen's inequality using the fact that $\sum_\ell \nu(\ell) \leq 1$ as follows.
\begin{subclaim}
\[\sum_\ell \nu(\ell) \cdot \left(\frac{\log\left(\frac{2^{d-1} \cdot S_\ell}{\nu(\ell)}\right)}{d}\right)^{dr}\leq \left[ \frac{\log\left(2^{d}\cdot S\right)}{d}\right]^{dr}.\]
\end{subclaim}
\begin{proof}
Recall that $\sum_\ell \nu(\ell) \leq 1$. Consider the random variable $Y$ defined as follows:
\newcommand{\bY}{\boldsymbol{Y}}
\begin{align*}
\bY \gets \begin{cases}
    \frac{2^{d-1} \cdot S_\ell}{\nu(\ell)} &\text{with probability } \nu(\ell) \text{ for each }\ell \text{ such that }\nu(\ell) \neq 0,\\
    1 &\text{with probability } 1- \sum_\ell \nu(\ell).
\end{cases}         
\end{align*}
and the concave function $x \overset{f}{\longmapsto} \left(\frac{\log x}{d}\right)^{dr}$. Applying Jensen's inequality, we obtain
\begin{align*}
    \E[f(\bY)] \leq f(\E \bY) &= \left[\frac{\log\left((\sum_\ell 2^{d-1}\cdot S_\ell) + (1-\sum_\ell \nu(\ell))\right)}{d}\right]^{dr}\\
    &\leq  \left[\frac{\log\left(2^{d-1}\cdot S + 1\right)}{d}\right]^{dr} &[\text{Since $S = \sum_\ell S_\ell$}]\\
    &\leq  \left[\frac{\log\left(2^{d}\cdot S \right)}{d}\right]^{dr} &[\text{Since $S \geq 1$}]
\end{align*}    
\end{proof}    
Plugging this bound back into our expression and recalling that $\lambda(F)= 32^{d+1}\cdot\left(\frac{\log 2^{d}\cdot S}{d}\right)^{d}$, we obtain
\begin{align*}
    \Pr\left[\CDT^{(a)}_0(F,\brhotilde) \text{ exists} \mid \brhotilde \in \calT\right] &\leq  \left(p \cdot \lambda(F)\right)^s \cdot \sum_{r} \frac1{2^r} \leq (p\cdot \lambda(F))^s,
\end{align*}
which completes the proof of our main lemma.
\end{proof}

\section{Satisfiablity algorithms}\label{sec:satisfiability}
In this section, we a give a randomized $\#SAT$ algorithm for general $\AC^0$ formulae, matching the Impagliazzo-Matthews-Paturi result for $\AC^0$ circuits. Rossman \cite{Rossman2019} had obtained a similar result for regular formulae. The proof below is a verbatim adaptation of Rossman's corresponding result \cite[Theorem 30]{Rossman2019} for regular formulae to the general setting. 
    \begin{theorem}
    \label{thm:AC0_SAT_algo}
    There is a randomized, zero-error algorithm which, given an $\AC^0$
    formula $F$ of depth $d + 1$ and size $S$ on n variables, outputs a decision tree for $F$ of size
    $O\left(Sn \cdot 2^{(1-\epsilon)n}\right)$ where $\epsilon = 1/ \bigo{\cbra{\frac1d \log S}^d}$. 
    This algorithm also solves the $\#SAT$ problem, that
    is, it counts the number of satisfying assignments for $F$.
    \end{theorem}
    \begin{proof}
        Given any depth $d$ formula, and restriction tree $\rhotilde$, the decision tree algorithm from \cref{def:CDT AC0} computes $\CDT(F,\rhotilde)$ in time $\bigo{n} \cdot \sum_{G \in T_F} \size(\CDT(G,\rhotilde|_{G}))$. 
        Given an $\AC^0$ formula, consider the following tree of subsets $\Dtilde\colon T_F \to [n]$ such that for each $G,H \in T_F$, such that $G$ is a parent of $H$,  $\Dtilde(H) \subseteq \Dtilde(G)$.  For each such $\Dtilde$, we get a decision tree for $F$ as follows: We first construct a decision tree $\Gamma$ by querying all the variables in $\Dtilde(F)$ and labelling each leaf with the corresponding restriction on $\Dtilde(F)$. For each such leaf $\sigma$ (i.e, for each choice $\sigma \colon \Dtilde(F) \to \zo$), we get a corresponding restriction tree $\rhotilde_{\Dtilde,\sigma}$ in the natural manner. For each such $\sigma$, construct $\CDT(F,\rhotilde_{\Dtilde,\sigma})$ and plug it in instead of the leaf corresponding to $\sigma$ in the complete binary tree $\Gamma$. Clearly, this resultant tree $\Gamma_{\Dtilde}$ is a decision tree for $F$. 
        
        We construct a (random) $\bGamma_{\bDtilde}$ by sampling a $\bDtilde$ as follows: randomly choose a $\btau \in_R [0,1]^n$ and set $\bDtilde(G):= \bra{i \colon \btau_i \leq 1 - 1/8\lambda(G)}$ for each $G \in T_F$. 
        Therefore the expected running time of the algorithm which computes the decision tree for $F$ is $O(n) \cdot {\sum_{G \in T_F} \E_{\btau}\left[\sum_{\sigma\colon \bDtilde(F) \to \zo}\size(\CDT(G,\brhotilde_{\bDtilde,\sigma}|_G))\right]}$ while the expected size of the decision tree is $\E_{\btau}\left[\sum_{\sigma\colon \bDtilde(F) \to \zo}\size(\CDT(F,\brhotilde_{\bDtilde,\sigma}))\right]$.
        
        We bound these expression as follows.       
        For each $G \in T_F$, size of the decision tree $\CDT(G,\rhotilde)$ is given by the expression, 
            \begin{align*}
                \E_{\btau}\left[\sum_{\sigma \colon \zo^{\bDtilde(F)}}\size(\CDT(F,\brhotilde_{\bDtilde,\sigma}))\right]
                & = \E_{\btau}\left[2^{|\bDtilde(F)|}\cdot \E_{\bsigma}\left[\size(\CDT(F,\brhotilde_{\bDtilde,\bsigma}))\right]\right]\\
                & \leq 2^{n(1-\nicefrac1{16\lambda(F)})}\cdot \underbrace{\E_{\btau,\bsigma}\left[\size(\CDT(F,\brhotilde_{\bDtilde,\bsigma}))\right]} + 2^n \cdot e^{-(\frac12)^2\cdot\frac12\cdot\frac{n}{8\lambda(G)}}
            \end{align*}
            where in the last expression, we have used the Chernoff Bound $\Pr[\sum X_i \leq (1-\delta)\mu]\leq e^{-\delta^2\mu/2}$ to bound the probability $\Pr[|[n] \setminus \bDtilde(F)| \leq (1-\nicefrac12)\mu]$ where $\mu = \E[|[n]\setminus \bDtilde(F)|]= \nicefrac{n}{8\lambda(F)}$. We can now further simplify the expression indicated in the underbraces as follows:
            \begin{align*}
                \E_{\btau,\bsigma}\left[\size(\CDT(F,\brhotilde_{\bDtilde,\bsigma}))\right]
                & = \sum_{t \geq 0} \sum_{a \in \zo^t}  \sum_{b \in \zo}\prob{\brhotilde_{\bDtilde,\bsigma}}{\CDT^{(a)}_b(F,\brhotilde_{\bDtilde,\bsigma}) \ \text{ exists}}\\
                & \leq 1 + \sum_{t =1}^{\infty} 2^t\left(\frac18\right)^t = \frac43.
            \end{align*}    

        We thus conclude that the expected size of the decision tree is at most $2^{n(1-\frac{1}{C\lambda})}$ for a suitably large constant $C$. 
    \end{proof}

\section*{Acknowledgements}

The first and the second authors spent several years thinking about this problem and we are indebted to several people along the way. First and foremost, we thank Jaikumar Radhakrishnan and Ramprasad Saptharishi for spending innumerable hours in the various stages of this project respectively going over various parts of the proof and giving us very helpful feedback. We are also greatly thankful to Ben Rossman both for initial discussions and pointing out an error in the previous version of this proof. In addition, we also thank Srikanth Srinivasan, Siddharth Bhandari, Yuval Filmus and Mrinal Kumar for their comments and feedback during the various stages of this project. 

{\small 
  \bibliographystyle{prahladhurl}
  \bibliography{HMS-bib.bib}
}

\end{document}